\documentclass[a4paper]{article}
\usepackage{fullpage}
\usepackage{graphicx}
\usepackage{float}
\usepackage{rotating}
\usepackage[font=small]{caption}
\usepackage{amsfonts}
\usepackage{amssymb,amsmath}

\usepackage{authblk}

\usepackage[font={small,it}]{caption}
\usepackage{wrapfig}
\usepackage{setspace}
\usepackage{natbib}
\setcitestyle{aysep={}} 
\usepackage{amsthm}
\newtheorem{Definition}{Definition}

\newtheorem{Proposition}{Proposition}
\newtheorem{Lemma}{Lemma}
\newtheorem{Algorithm}{Algorithm}
\usepackage{mathtools}
\DeclarePairedDelimiter{\ceil}{\lceil}{\rceil}
\newcommand\numberthis{\addtocounter{equation}{1}\tag{\theequation}}
\usepackage{marvosym}
\usepackage{subcaption}

\begin{document}
\title{}
\date{}
\author{}
\noindent
{\LARGE\bf Co-modularity and Detection of Co-communities}
\vspace{1ex}
\\
Thomas E. Bartlett$^{1\ast}$
\vspace{1ex}
\\
1. Department of Statistical Science, University College London, London WC1E 7HB, United Kingdom.
\\
$\ast$ E-mail: thomas.bartlett.10@ucl.ac.uk
\begin{abstract}
This paper introduces the notion of co-modularity, to co-cluster observations of bipartite networks into co-communities. The task of co-clustering is to group together nodes of one type with nodes of another type, according to the interactions that are the most similar. The novel measure of co-modularity is introduced to assess the strength of co-communities, as well as to arrange the representation of nodes and clusters for visualisation, and to define an objective function for optimisation. We demonstrate the usefulness of our proposed methodology on simulated data, and with examples from genomics and consumer-product reviews.
\end{abstract}


\section{Introduction}\label{introSect}
Networks are used to parsimoniously represent relationships between entities of the same type. Classical analysis methods use parametric models of network data, such as degree-based and/or community-based models \citep{holland1983stochastic,bickel2009nonparametric,rohe2011spectral,qin2013regularized,wilson2013testing,zhang2015estimating}. The last few years have seen a flurry of activity in statistical network analysis, see for example \citep{amini2013pseudo,cai2014robust,newman2016equivalence,zhang2020statistical}. One of the best-studied tools is the stochastic blockmodel \citep{holland1983stochastic,bickel2009nonparametric}, and various extensions to it, \citep{zhao2012consistency,airoldi2009mixed,gopalan2013efficient}; various methods of fitting this model have been proposed, where maximizing modularity remains an important practical approach \citep{girvan2002community,bickel2009nonparametric}. Recent work in clustering network nodes has generalised the applicability of the stochastic blockmodel, by showing that arbitrary exchangeable networks can be represented using a blockmodel \citep{diaconis1977finite,bickel2009nonparametric,olhede2014net}; such a representation is called a `network histogram'. The network histogram method \citep{olhede2014net} can be used to estimate the optimal granularity at which communities, or functional subnetwork modules, can be approximated and isolated in social and biological networks; i.e., to estimate the optimal number of clusters or communities of network nodes. Alternatively, several Bayesian approaches to estimating the optimal number of communities in a network have also been proposed \citep{riolo2017efficient,peixoto2019bayesian,yen2020community}. However, it is well established that when clustering is implemented, estimating the optimal number of clusters is an important and separate problem from the design of the clustering methodology \citep{bickel2009springer}. For example, sophisticated solutions to this problem such as the gap statistic \citep{tibshirani2001estimating} propose methodology for estimating the optimal number of clusters, and this is done independently from the choice of clustering methodology. In this paper we focus on the problem of clustering methodology for variables of different types, i.e., co-clustering.

Studying relationships between variables of the same type is naturally very useful; its simplest generalisation is to study relationships between variables of a different type; this is known as the co-clustering problem \citep{flynn2012consistent,choi2014co,madeira2004biclustering,gao2016optimal}, and is of much current interest in application areas from genomics to natural language processing \citep{bhattacharya2017gpu,clevert2017rectified,rugeles2017biclustering}. The co-clustering problem can also be approached non-parametrically, as is made clear by \cite{choi2014co} and \cite{gao2016optimal}. We start from the modularity approach to recognising communities \citep{girvan2002community}, realising that extending such understanding to variables of different types is nontrivial \citep{aldous1985exchangeability,madeira2004biclustering}. Having recognised communities in both types of variables, we need to transform the clustering or grouping of both types of variables into an ordering of groups. This is not inherent to the formulation of the Aldous-Hoover representation of the generating mechanism of the random array we are modelling, but is important for visualisation purposes. We aggregate the modularity over the groupings of each of the two types of variables (corresponding to {\em rows} or {\em columns}) to guide this choice of visualisation. We also use modularity to compare co-communities, which we define as pairings of one group or cluster of each type of node, providing a unique paradigm for understanding these important bipartite network structures. Finally, to demonstrate the usefulness of our proposed methodology, we analyse two characteristic network data sets, as follows. A genomics data set, in which a co-community represents a functional module that involves two types of genomic features (i.e., measured on two different data modalities or platforms). And, a movie review data set, in which a co-community represents a set of movies (probably with similar characteristics) enjoyed by a particular group of (possibly similar) people. Importantly, in our model set-up, it is possible for a node to be part of one co-community, or part of multiple co-communities, or part of no co-communities. Thus, we show how our proposed analysis methods enable us to discover both known and hitherto unknown characteristics of these two data sets. This paper is organised as follows: Section \ref{coModCoComDet} defines the stochastic block model, and gives the representation of an arbitrary separately exchangeable array. It also defines the co-modularity, and explains how the array data will be analysed. Then Section \ref{idenCompCoCom} shows how to find the co-communities in data, and Section \ref{examples} gives examples to illustrate the performance of our proposed methodology. The Appendix provides all proofs of the paper. 

\section{Co-modularity and Co-Community Detection}\label{coModCoComDet}
\begin{wrapfigure}{r}{0.35\textwidth}
\vspace{-6ex}
\begin{center}
\includegraphics[width=0.33\textwidth]{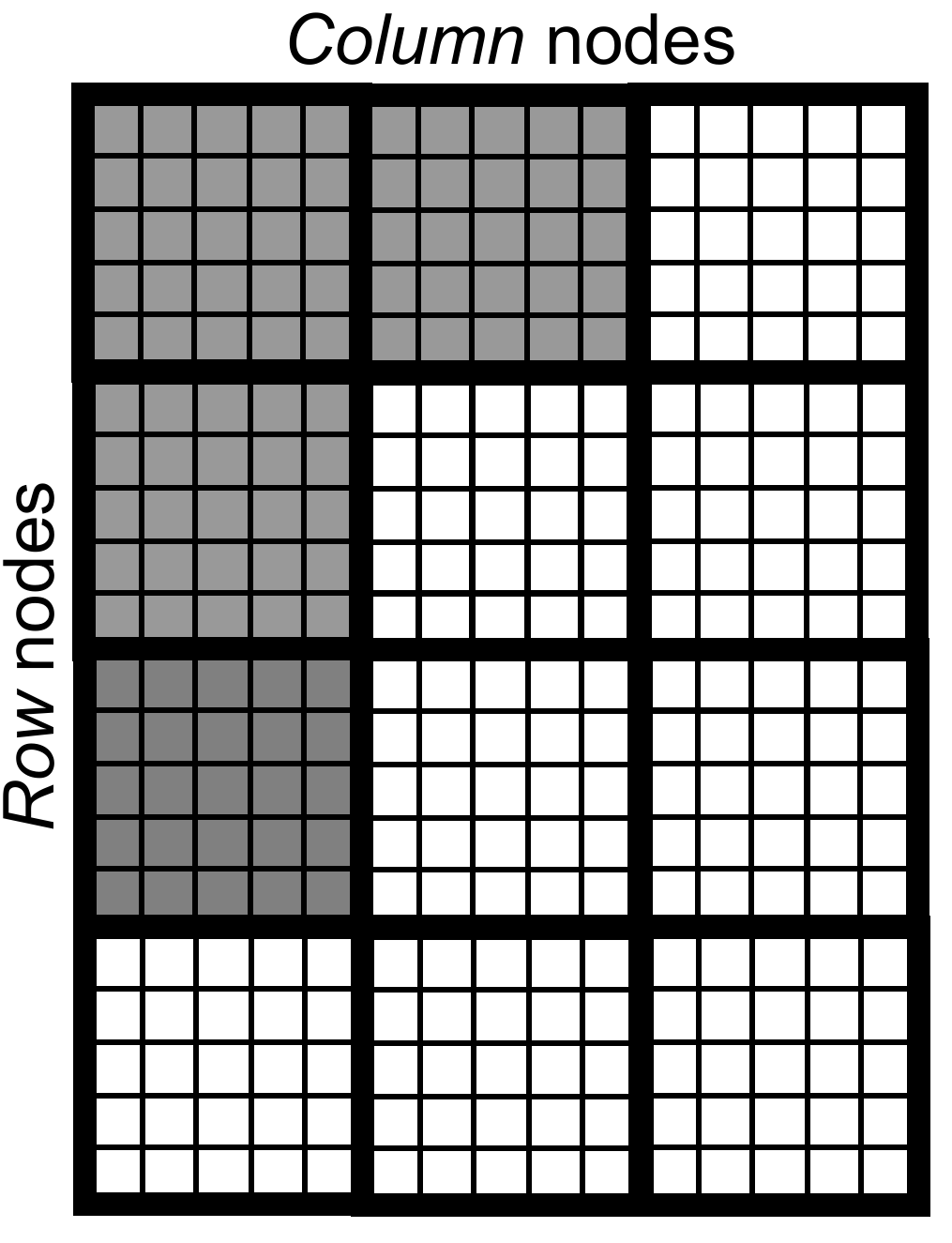}
\end{center}
\vspace{-3ex}
\caption{Co-community structure.} \label{coComBaseEg}
\caption*{The margins of the random array represent different types of nodes, and the array elements define the edges of the bipartite network. The shaded blocks represent co-communities of the two different types of nodes.}
\vspace{-2ex}
\end{wrapfigure}
We begin this section by defining the degree-corrected stochastic co-blockmodel \citep{rohe2012co,flynn2012consistent,choi2014co} together with notation. We then give a definition of the Newman-Girvan modularity \citep{newman2004finding}, and by analogy, we define a quantity which we term the `co-modularity', and we specify an algorithm for maximising this quantity. Whilst previous work (see e.g. that by \cite{dhillon2001co}) has separately identified groupings of the different types of nodes, the notion of `co-modularity' considers the pairing into blocks of these groupings of different types of nodes (an issue which doesn't arise in networks with only one type of node). We call such a pairing a `co-community', illustrated by the shaded blocks in Fig.\ref{coComBaseEg}. We show that under certain conditions, maximising the co-modularity in this way is equivalent to maximising the model likelihood of the specified degree corrected stochastic co-blockmodel.
\begin{Definition}[Degree-corrected stochastic co-blockmodel]\label{mainModelDef}
For $m,l\in\mathbb{N}^+$, assign the labelling for the set of $X$-nodes as $\{1,...,m\}$, and for the set of $Y$-nodes $\{1,...,l\}$, where this labelling is chosen without loss of generality. Denote an $X$-node grouping as $g^{(X)}_p\in G^{(X)}$, $p\in\left\{1,...,k^{(X)}\right\}$, and a $Y$-node grouping as $g^{(Y)}_q\in G^{(Y)}$, $q\in\left\{1,...,k^{(Y)}\right\}$, where $G^{(X)}$ and $G^{(Y)}$ are exhaustive lists of mutually exclusive $X$ and $Y$-node groupings, respectively. Define map functions $z^{(X)}(i)$ and $z^{(Y)}(j)$, such that $g^{(X)}_p=\left\{i:z^{(X)}(i)=p\right\}$, and $g^{(Y)}_q=\left\{j:z^{(Y)}(j)=q\right\}$. Define co-community connectivity parameters $\boldsymbol{\theta}\in[0,1]^{k^{(X)}\times k^{(Y)}}$, where $\theta_{z^{(X)}(i),z^{(Y)}(j)}$ is the propensity of $X$-node $i$ in group $z^{(X)}(i)$ to form a connection with $Y$-node $j$ in group $z^{(Y)}(j)$. Define also node-specific connectivity parameters $\boldsymbol{\pi}^{(X)}\in\mathbb{R}_{\ge 0}^m$ and $\boldsymbol{\pi}^{(Y)}\in\mathbb{R}_{\ge 0}^l$. Let the elements of the adjacency matrix $\mathbf{A}\in\{0,1\}^{m\times l}$ follow the law of:
\begin{equation}
A_{ij}\sim\operatorname{Bernouilli}\left({{\pi}^{(X)}_i{\pi}^{(Y)}_j\theta_{z^{(X)}(i),z^{(Y)}(j)}}\right), 1\leq i\leq m, 1\leq j\leq l.\label{DCSCBM}
\end{equation}
Then, we call the generative mechanism of $A_{ij}$ `the degree corrected stochastic co-blockmodel'.
\end{Definition}
\noindent
We note that the terminology `$X$-nodes' and `$Y$-nodes' is non-standard; we introduce it here, to increase clarity of exposition. To improve identifiability of parameters of the model in Definition \ref{coModCoComDet}, and defining a co-community as a pairing of the $X$-node grouping $g^{(X)}_p$ with the $Y$-node grouping $g^{(Y)}_q$, we introduce a special case of the blockmodel favoured by many other authors \citep{newman2013spectral}, that $\theta_{z^{(X)}(i),z^{(X)}(j)}$ may take only two values: 
\begin{align}
   \theta_{p,q}=&\begin{cases}
    \theta_{\mathrm{in}}, & \text{if the pairing of $X$-node grouping }g^{(X)}_p\text{ with $Y$-node}\\
    							 & \text{grouping }g^{(Y)}_q\text{ is a co-community},\\
    \theta_{\mathrm{out}}, & \text{otherwise}.
\end{cases} \label{poissonCons}
\end{align}
We can also replace the Bernoulli model likelihood with a Poisson likelihood: because the Bernoulli success probability is typically small, and the number of potential edges (i.e., pairings of nodes) is large, a Poisson distribution with the same mean converges to the same distribution, and so it makes little difference in practice \citep{zhao2012consistency,perry2012null}. Its usage greatly simplifies the technical derivations. Hence, we calculate the model log-likelihood as follows (assuming $A_{ij}\in\{0,1\}$ and therefore $A_{ij}!=1$ for all $i$, $j$):
\begin{multline}
\ell\left(\boldsymbol{\theta},\boldsymbol{\pi}^{(X)},\boldsymbol{\pi}^{(Y)};G^{(X)},G^{(Y)}\right)\\
=\sum_{i=1}^m\sum_{j=1}^lA_{ij}\ln{\left(\pi^{(X)}_i\pi^{(Y)}_j\theta_{z^{(X)}(i),z^{(Y)}(j)}\right)}-\pi^{(X)}_i\pi^{(Y)}_j\theta_{z^{(X)}(i),z^{(Y)}(j)}. \label{modelEll}
\end{multline}

The Newman-Girvan modularity \citep{newman2004finding} measures, for a particular partition of a network into communities, the observed number of edges between community members, compared to the expected number of edges between community members without the community partition with the degree correction. The Newman-Girvan modularity may be defined as follows:
\begin{Definition}[Newman-Girvan modularity]\label{NGmodDef}
Define $\mathbf{A}\in\{0,1\}^{n\times n}$ as a symmetric adjacency matrix representing a unipartite network with nodes $i\in\left\{1,...,n\right\}$, define $\mathbf{d}$ as the degree vector of the nodes of this network, $d_i=\sum_{j=1}^nA_{ij}$, and define the normalising factor $d^{++}$ as twice the total number of edges, $d^{++}=\sum_{i=1}^nd_i$. Define a community, or grouping, of nodes as $g\in G$, where $G$ represents the set of all such groupings of nodes, define the map function $z(i)$ such that $g_a=\left\{i:z(i)=a\right\}$, and let $\mathbb{I}\left[z(i)=z(j)\right]$ specify whether nodes $i$ and $j$ appear together in any community $g$, such that:
\[
    \mathbb{I}\left[z(i)=z(j)\right]= 
\begin{cases}
    1, & \text{if nodes }i\text{ and }j\text{ are grouped together}\\ 
    	& \text{in any community }g\in G,\\
    0, & \text{otherwise}.
\end{cases}
\]
Then, the Newman-Girvan modularity $Q_{NG}$ is defined as:
\begin{align}
Q_{NG}=\frac{1}{d^{++}}\sum_{i=1}^n\sum_{j=1}^n\left[A_{ij}-\frac{d_id_j}{d^{++}}\right]\cdot\mathbb{I}\left[z(i)=z(j)\right]. \label{NGmod}
\end{align}
\end{Definition}
\noindent
The co-modularity is then defined by analogy with the Newman-Girvan modularity (Definition \ref{NGmodDef}) as follows:
\begin{Definition}[Co-modularity]\label{coModDef}
With $m$ and $l$ given by Definition \ref{mainModelDef} and $\mathbf{A}$ generated according to Definition \ref{mainModelDef}, define $\mathbf{d}^{(X)}$ and $\mathbf{d}^{(Y)}$ as the degree vectors of the $X$ and $Y$-nodes of the network, $d^{(X)}_i=\sum_{j=1}^lA_{ij}$ and $d^{(Y)}_j=\sum_{i=1}^mA_{ij}$, and define the normalising factor $d^{++}$ as twice the total number of edges, $d^{++}=\sum_{i=1}^md^{(X)}_i=\sum_{j=1}^ld^{(Y)}_j$. With $g^{(X)}$ and $g^{(Y)}$, $z^{(X)}$ and $z^{(Y)}$ also defined in direct analogue according to Definition \ref{mainModelDef}, let $c_t=\left\{p,q\right\}\in C$, $t=\left\{1,...,T\right\}$. The enumeration of the pair $\left\{p, q\right\}$ is arbitrary, and is to facilitate ease of access of the co-blocks in a chosen order. The co-block $c_t$ specifies that the $X$-node grouping $g^{(X)}_p$ is paired with the $Y$-node grouping $g^{(Y)}_q$; we refer to such a pairing as a `co-community'. Furthermore, let $\Psi\left(C;G^{(X)},G^{(Y)};i,j\right)\in\left\{0,1\right\}$ specify whether nodes $i$ and $j$ appear together in any co-community $c\in C$, such that:
\[
    \Psi\left(C;G^{(X)},G^{(Y)};i,j\right)= 
\begin{cases}
    1, & \text{if }\left\{z^{(X)}(i),z^{(Y)}(j)\right\}=c:c\in C,\\
    0, & \text{otherwise}.
\end{cases}
\]
Then, the co-modularity $Q_{XY}$ is defined as:
\begin{align}
Q_{XY}=\frac{1}{d^{++}}\sum_{i=1}^m\sum_{j=1}^l\left[A_{ij}-\frac{d^{(X)}_id^{(Y)}_j}{d^{++}}\right]\Psi\left(C;G^{(X)},G^{(Y)};i,j\right). \label{coMod}
\end{align}
\end{Definition}
\noindent
We note that for the co-modularity (unlike the Newman-Girvan modularity), we require a set of pairings of $X$-node groupings with $Y$-node groupings $C$, such that each $c_t\in C$ is a pairing of an $X$-node grouping $g^{(X)}_p\in G^{(X)}$ with a $Y$-node grouping $g^{(Y)}_q\in G^{(Y)}$. Also, due to the asymmetry of the co-clustering problem, $c_t=\left\{p,q\right\}\neq\left\{q,p\right\}$. This separately specified set of parings $C$ is not required in the case of the Newman-Girvan modularity, because in the unipartite network setting, there is only one type of node, and hence node groupings already `match-up' with one another. This can be visualised, in the unipartite network setting, as community structure present along the leading diagonal of the adjacency matrix, if the nodes are ordered by community. In the co-community setting, an $X$-node grouping $g^{(X)}$ may be paired in $C$ with many, with one, or with no $Y$-node groupings $g^{(Y)}\in G^{(Y)}$, and equivalently a $Y$-node grouping $g^{(Y)}$ may be paired in $C$ with many, with one, or with no $X$-node groupings $g^{(X)}\in G^{(X)}$ (Fig.\ref{coComBaseEg}). Further, if the $X$-nodes and $Y$-nodes of the network are arranged in the adjacency matrix according to the groupings $g^{(X)}$ and $g^{(Y)}$, there is no reason co-communities should appear along the leading diagonal. Hence, the function $\Psi$ in Equation \ref{coMod} generalises the role of the indicator function in Equation \ref{NGmod}. We note that other approaches to this problem directly specify a null model \citep{bassett2013robust}. We also note that sometimes in practice, we must relax the requirement of $\Psi\in\left\{0,1\right\}$; the reason for this is made clear in the technical derivations (for tractability) in Appendix A which relate to Algorithm \ref{specClustAlg} (which follows next).

Community detection of $k$ communities can be performed by fitting the degree-corrected stochastic blockmodel. This is equivalent, under many circumstances, to spectral clustering \citep{bickel2009nonparametric,riolo2014first,newman2013spectral}, which may be carried out by grouping the nodes into $k$ clusters in the space of the eigenvectors corresponding to the 2\textsuperscript{nd} to $k$\textsuperscript{th} greatest eigenvalues of the Laplacian $\mathbf{L}=\mathbf{D}^{-1/2}\mathbf{A}\mathbf{D}^{-1/2}$, where $\mathbf{D}$ is the diagonal matrix of the degree distribution. Co-community detection in a bipartite network of nodes attributed to the variables $X$ and $Y$ (respectively, $X$-nodes and $Y$-nodes), can equivalently be performed by degree-corrected spectral clustering \citep{dhillon2001co}. 

A procedure to find an assignment of $X$ and $Y$-nodes to $k^{(X)}$ $X$-node groupings (`row clusters') and  $k^{(Y)}$ $Y$-node groupings (`column clusters') respectively, which finds a (possibly locally) optimum value of the co-modularity $Q_{XY}$, is specified in Algorithm \ref{specClustAlg}:
\begin{Algorithm}\label{specClustAlg}
With $\mathbf{A}$ and $Q_{XY}$ defined as in Definition \ref{mainModelDef}, and $\mathbf{d}^{(X)}$ and $\mathbf{d}^{(Y)}$ defined as in Definition \ref{coModDef}:
\begin{enumerate}
\item Calculate the co-Laplacian $\mathbf{L}_{XY}$ \citep{dhillon2001co} as:
\begin{equation}
\mathbf{L}_{XY}=\left(\mathbf{D}^{(X)}\right)^{-1/2}\mathbf{A}\left(\mathbf{D}^{(Y)}\right)^{-1/2}, \label{coLaplacian}
\end{equation}
where $\mathbf{D}^{(X)}$ and $\mathbf{D}^{(Y)}$ are the diagonal matrices of $\mathbf{d}^{(X)}$ and $\mathbf{d}^{(Y)}$, respectively.
\item Calculate the singular value decomposition (SVD) of the co-Laplacian $\mathbf{L}_{XY}$.
\item Separately cluster the $X$ and $Y$-nodes in the spaces of the left and right singular vectors corresponding to the 2\textsuperscript{nd} to $k^{(X)}$\textsuperscript{th} and 2\textsuperscript{nd} to $k^{(Y)}$\textsuperscript{th} greatest singular values, respectively, of this SVD of $\mathbf{L}_{XY}$.
\item Identify the set of co-communities $C$, as pairings of particular $X$-node groupings and $Y$-node groupings.
\end{enumerate}
\end{Algorithm}
\noindent
We note that there are a range of choices and alternatives for steps 1-3 of Algorithm \ref{specClustAlg}. The specifics of how to identify the set of co-communities in step 4 of Algorithm \ref{specClustAlg}, i.e., the identification of $C$ as particular pairings of the identified $X$-node groupings and $Y$-node groupings, are discussed further in Section \ref{idenCompCoCom}. Technical derivations relating to Algorithm \ref{specClustAlg} appear in Appendix A, and are based on arguments made previously in the context of unipartite (symmetric) community detection \citep{newman2013spectral} for two communities, extending them to this context of (asymmetric) co-community detection. We note in particular, that the notion of modularity assumes that within-community edges are more probable than between-community edges, and therefore modularity maximisation is only consistent if constraints are applied to ensure this assumption holds \citep{zhao2012consistency}. In the community detection setting, under suitable constraints, the solutions which maximise model likelihood and modularity are identical \citep{bickel2009nonparametric}.
\begin{Proposition}\label{poisLprop}
The solution which maximises the model likelihood specified in equation \ref{modelEll}, subject also to the constraint of equation \ref{poissonCons}, is equivalent to the maximum co-modularity assignment obtained via Algorithm \ref{specClustAlg}.
\end{Proposition}
\begin{proof}
The proof for the case of two co-communities appears in Appendix B. It extends arguments made previously in relation to community detection \citep{newman2013spectral} to this context of co-community detection.
\end{proof}

\section{Identification and Comparison of Co-communities}\label{idenCompCoCom}
Fitting the stochastic co-blockmodel by spectral clustering as described in Algorithm \ref{specClustAlg}, involves using $k$-means to cluster the $X$ and $Y$-nodes in the spaces of the left and right singular vectors of the co-Laplacian (equation \ref{coLaplacian}). However, there is subsequently the problem of how to identify the co-communities, as represented by the shaded blocks in Fig.\ref{coComBaseEg}. This problem of identifiability is novel: it does not arise when fitting the stochastic blockmodel to unipartite networks by spectral clustering, because of the symmetry of the problem (if there is only one type of node, then different types of nodes do not need to be grouped). Finding the co-communities (illustrated by the shaded blocks in Fig.\ref{coComBaseEg}) consists of estimating the set $C$ (Definition \ref{coModDef}) of pairings of $X$-node groupings $g^{(X)}\in G^{(X)}$ with $Y$-node groupings $g^{(Y)}\in G^{(Y)}$. It replaces identifying the ``diagonal'' of the blockmodel, a concept that is less straightforward than for a symmetric adjacency matrix. In this section we propose novel methodology to address this problem in a fully automated way in the bipartite network setting, along with related issues of visualisation and optimisation.

Fitting the symmetric blockmodel in the unipartite community detection setting, there are exactly $k=k^{(X)}=k^{(Y)}$ communities (because of symmetry). Each row grouping matches up with exactly one column grouping, because the row and column groupings are the same thing. On the other hand, fitting the asymmetric co-blockmodel by spectral clustering as in Algorithm \ref{specClustAlg} leads to $k^{(X)}$ and $k^{(Y)}$ row and column clusters. Hence, these $k^{(X)}$ and $k^{(Y)}$ row and column clusters provide $k^{(X)}\times k^{(Y)}$ potential co-communities. Which of these are significant (as in the shaded blocks in Fig.\ref{coComBaseEg})? The best-known solution to this problem clusters all the nodes at once (after concatenating the left and right singular vectors) \citep{dhillon2001co}, instead of clustering the $X$ and $Y$-nodes separately. However, that approach requires $k^{(X)}=k^{(Y)}$, so that each $X$ node grouping is paired with exactly one $Y$-node grouping. Our approach does not have this restriction, and hence can model a broader class of bipartite network structures, by introducing the notion of co-communities, as illustrated in Fig.\ref{coComBaseEg}. On the other hand, our approach must answer the question, how should we assess and compare the $k^{(X)}\times k^{(Y)}$ potential co-communities? I.e., how should we compare each different pairing of an estimated $X$-node grouping $\hat{g}^{(X)}\in\hat{G}^{(X)}$, with an estimated $Y$-node grouping $\hat{g}^{(Y)}\in\hat{G}^{(Y)}$, to provide an assignment of the  $X$-nodes and $Y$-nodes to co-communities? In practice, we expect the number of co-communities, $T=\left|C\right|$ (where $\left|\cdot\right|$ represents cardinality), to be significantly less than $k^{(X)}\times k^{(Y)}$. In the unipartite community detection setting, $k^{(X)}=k^{(Y)}=k$, and only the blocks on the diagonal can be communities: hence in effect there we have $T=k=\sqrt{k^{(X)}\times k^{(Y)}}$.

To estimate the set of co-communities, $c_t\in C$, $t=\left\{1,...,T\right\}$, in this bipartite network setting, we calculate the `local co-modularity' for each pairing $\hat{g}^{(X)}$ with $\hat{g}^{(Y)}$, by considering a relevant sub-part of the co-modularity matrix $\textbf{B}$ (equation \ref{Bdef}):
\begin{Definition}[Local co-modularity]\label{localCoModDef}
With $\mathbf{A}$ given by Definition \ref{mainModelDef}, with $\mathbf{d}^{(X)}$, $\mathbf{d}^{(Y)}$ and $d^{++}$ given by Definition \ref{coModDef}, with 
\begin{equation}
B_{ij}=A_{ij}-\frac{d^{(X)}_id^{(Y)}_j}{d^{++}},\qquad
\mathbf{B}=\mathbf{A}-\frac{1}{d^{++}}\mathbf{d}^{(X)}\left(\mathbf{d}^{(Y)}\right)^{\top}, \label{Bdef}
\end{equation}
and with the set of $X$-node groupings and the set of $Y$-node groupings estimated according to Algorithm \ref{specClustAlg} as $\hat{G}^{(X)}$ and $\hat{G}^{(Y)}$ respectively, where $\left|\hat{G}^{(X)}\right|=k^{(X)}$ and $\left|\hat{G}^{(Y)}\right|=k^{(Y)}$, where $\left|\cdot\right|$ represents cardinality, for a particular pairing of estimated $X$-node grouping $\hat{g}^{(X)}\in\hat{G}^{(X)}$ with estimated $Y$-node grouping $\hat{g}^{(Y)}\in\hat{G}^{(Y)}$, the local co-modularity $Q_{XY}\left(\hat{g}^{(X)},\hat{g}^{(Y)}\right)$ is defined as:
\begin{equation}
Q_{XY}\left(\hat{g}^{(X)},\hat{g}^{(Y)}\right)=\frac{1}{d^{++}}\sum_{i\in \hat{g}^{(X)}}\sum_{j\in \hat{g}^{(Y)}}B_{ij}. \label{localCoMod}
\end{equation}
\end{Definition}
\noindent `Local' here means that we are considering a statistic for an individual block, out of the many blocks which are found in general along each row and column. Each of the $k^{(X)}\times k^{(Y)}$ possible pairings of $\hat{g}^{(X)}$ with $\hat{g}^{(Y)}$ can be defined, or not, as a co-community; doing so means that they are included in, or excluded from, the estimated set of co-communities $\hat{C}$ (Definition \ref{coModDef}). To consider all permutations, $2^{k^{(X)}\times k^{(Y)}}$ such assignments would need to be considered, which would be computationally very demanding. However, this problem can be avoided by defining summary statistics targeted for particular purposes. The three such purposes which we consider here are described in the following subsections: \ref{compAssessSigCoCom} Comparing potential co-communities and assessing their strength; \ref{arrCoComVis} Arranging the co-communities for visualisation; \ref{DefObjFunOptCoCom} Defining an algorithmic objective function to be optimised, when determining co-community partitions.

\subsection{Comparing and assessing significance of co-communities}\label{compAssessSigCoCom}
Under a null model of no co-community structure, $\theta_{z^{(X)}(i),z^{(Y)}(j)}=\text{constant}$, for all $i, j$. Therefore, referring to the log-linear model \citep{perry2012null}, equation \ref{DCSCBM} becomes:
\begin{equation}
A_{ij}\sim\operatorname{Bernouilli}\left(\frac{{\pi}^{(X)}_i{\pi}^{(Y)}_j}{\pi^{++}}\right)\label{DCSCBMnull},
\end{equation}
where we have defined:
\begin{equation}
\theta_{z^{(X)}(i),z^{(Y)}(j)}=1/\pi^{++}. \label{nullAssump}
\end{equation}
Hence under this null,
\begin{equation*}
\mathbb{E}\left(A_{ij}\right)=\frac{{\pi}^{(X)}_i{\pi}^{(Y)}_j}{{\pi}^{++}},
\end{equation*}
which implies that for large networks which are not too sparse, $\mathbb{E}\left(B_{ij}\right)$ is nearly zero. We define the informal idealised quantities $\widetilde{\mathbf{B}}$ and $\widetilde{Q}_{XY}$ in comparison with equations \ref{Bdef} and \ref{localCoMod}:
\begin{equation}
\widetilde{\mathbf{B}}=\mathbf{A}-\frac{1}{{\pi}^{++}}{\boldsymbol{\pi}}^{(X)}\left({\boldsymbol{\pi}}^{(Y)}\right)^{\top},
\end{equation}
and
\begin{equation}
\widetilde{Q}_{XY}\left(\hat{g}^{(X)},\hat{g}^{(Y)}\right)=\frac{1}{{\pi}^{++}}\sum_{i\in \hat{g}^{(X)}}\sum_{j\in \hat{g}^{(Y)}}\widetilde{B}_{ij},\label{theorLocalCoMod}
\end{equation}
where the empirical degree distributions $\mathbf{d}^{(X)}$ and $\mathbf{d}^{(Y)}$ have been replaced by the theoretical node connectivity parameters $\boldsymbol{\pi}^{(X)}$ and $\boldsymbol{\pi}^{(Y)}$, and the empirical normalisation factor $d^{++}$ is also replaced by the theoretical normalisation factor $\pi^{++}$. 

If the pairing of $X$ and $Y$-node groupings $\hat{g}^{(X)}$ and $\hat{g}^{(Y)}$ exhibit some co-community structure, then equation \ref{nullAssump} no longer holds, and so the null model does not hold either. The stronger this co-community structure is, the further we move from the null model, and the greater $\theta$ becomes relative to $1/\pi^{++}$. This corresponds to $\mathbb{E}\left(A_{ij}\right)$ becoming larger than $\pi^{(X)}_i\pi^{(Y)}_j/\pi^{++}$, which is equivalent to the observed number of edges in the co-community becoming greater than the expected, under the null of no co-community structure. This in turn means that $\widetilde{Q}_{XY}$ also becomes more positive. In other words, the further we move from the null model, the greater tendency of the $X$-nodes and $Y$-nodes of these groups to form connections with one another (compared with their expected propensity to make connections with any nodes, of the opposite type), and therefore constitute a strong co-community. Hence, a parsimonious method of comparing potential co-communities is simply to compare their local co-modularity, $Q_{XY}\left(\hat{g}^{(X)},\hat{g}^{(Y)}\right)$. This naturally leads to a ranking of potential co-communities according to their strength. e.g., leading to identification of the shaded blocks in Fig.\ref{coComBaseEg} by some thresholding criterion, such as statistical significance.

An estimate of statistical significance of a potential co-community can also be made, as follows. Noting that, with adjacency matrix $\mathbf{A}$ defined according to the Bernoulli distribution of Definition \ref{mainModelDef}, with fixed $\theta_{z^{(X)}(i),z^{(Y)}(j)}=1/{\pi}^{++}$,
\begin{equation*}
\mathrm{Var}\left(\widetilde{B}_{ij}\right)=\mathrm{Var}\left(A_{ij}\right)=\left(\frac{{\pi}^{(X)}{\pi}^{(Y)}}{{\pi}^{++}}\right)\left(1-\frac{{\pi}^{(X)}{\pi}^{(Y)}}{{\pi}^{++}}\right),
\end{equation*}
and assuming probabilities of observing links between different pairs of nodes are independent, the variance of $\widetilde{Q}_{XY}\left(\hat{g}^{(X)},\hat{g}^{(Y)}\right)$ can be approximated (for deterministic node-groupings) as:
\begin{equation}
\mathrm{Var}\left(\widetilde{Q}_{XY}\left(\hat{g}^{(X)},\hat{g}^{(Y)}\right)\right)=\frac{1}{\left({\pi}^{++}\right)^2}\sum_{i\in \hat{g}^{(X)}}\sum_{j\in \hat{g}^{(Y)}}\left(\frac{{\pi}^{(X)}_i{\pi}^{(Y)}_j}{{\pi}^{++}}\right)\left(1-\frac{{\pi}^{(X)}_i{\pi}^{(Y)}_j}{{\pi}^{++}}\right),\label{coModVarEst}
\end{equation}
where the factor $1/\left({\pi}^{++}\right)^2$ is due to the factor $1/\left({\pi}^{++}\right)$ in equation \ref{theorLocalCoMod}. We also note that here, departing from convention, the $\pi^{(X)}_i$ and $\pi^{(Y)}_j$ are of the scale of degrees rather than probabilities (Definition \ref{mainModelDef}). Hence, assuming $\mathbf{d}^{(X)}\overset{p}{\to}{\boldsymbol{\pi}}^{(X)}$, $\mathbf{d}^{(Y)}\overset{p}{\to}{\boldsymbol{\pi}}^{(Y)}$ and $d^{++}\overset{p}{\to}{\pi}^{++}$, and assuming the potential co-community defined by $\hat{g}^{(X)}$ and $\hat{g}^{(Y)}$ is comprised of sufficiently many nodes for a Gaussian approximation to hold, we can test the significance of $Q_{XY}\left(\hat{g}^{(X)},\hat{g}^{(Y)}\right)$ with a $z$-test, with zero mean and with $\mathrm{Var}\left({Q}_{XY}\right)$ estimated as $\mathrm{Var}\left(\widetilde{Q}_{XY}\right)$ in equation \ref{coModVarEst}, also replacing ${\pi}^{(X)}_i$ with $d^{(X)}_i$, ${\pi}^{(Y)}_j$ with $d^{(Y)}_j$ and ${\pi}^{++}$ with $d^{++}$, and ignoring the stochasticity of $g^{(X)}$ and $g^{(Y)}$. A pairing $\hat{g}^{(X)}_p$ and $\hat{g}^{(Y)}_q$ is then defined as a co-community $\hat{c}$ and included in $\hat{C}$ (Definition \ref{coModDef}), i.e., $\left\{p,q\right\}=\hat{c}\in\hat{C}$, if and only if this pairing $\hat{g}^{(X)}_p$ with $\hat{g}^{(Y)}_q$ is significant according to this $z$-test, at some significance level. We note that, in practice, this is only a rough approximation of significance, also because by specifying in advance the co-community node-groupings $\hat{g}^{(X)}$ and $\hat{g}^{(Y)}$, we have introduced dependencies between the $X$ and $Y$-nodes of this co-community.

\subsection{Arranging the co-communities for visualisation}\label{arrCoComVis}
A standard task in exploratory data analysis using variants of the stochastic block model, is arranging the detected communities so they can be visualised in a helpful way. This visualisation is usually carried out by way of a heatmap representation of the adjacency matrix with the nodes grouped into communities. In the symmetric/unipartite community detection scenario, the communities occur along the leading diagonal of this ordered adjacency matrix. The communities themselves are often ordered along the leading diagonal according to their edge densities. In the bipartite co-community detection setting, co-communities may be present away from the leading diagonal, and there is no longer a restriction on how many co-communities a node may be part of - although we do not consider here the possibility of overlapping co-communities. 

We propose then, that once the $X$-node groupings and $Y$-node groupings have been determined by spectral clustering as described above, a natural way to order these groups with respect to one another, is via row and column co-modularities, which we define as follows.
\begin{Definition}\label{rcCoModDef}
With $d^{++}$ given by Definition \ref{coModDef}, and with $\mathbf{B}$ given by Definition \ref{localCoModDef}, with with the set of $X$-node groupings and the set of $Y$-node groupings estimated according to Algorithm \ref{specClustAlg} as $\hat{G}^{(X)}$ and $\hat{G}^{(Y)}$ respectively, the row and column modularities $Q_{row}\left(\hat{g}^{(X)}\right)$ and $Q_{column}\left(\hat{g}^{(Y)}\right)$ are defined, for $\hat{g}^{(X)}\in\hat{G}^{(X)}$ and $\hat{g}^{(Y)}\in\hat{G}^{(Y)}$, as:
\begin{equation}
Q_{row}\left(\hat{g}^{(X)}\right)=\sum_{\hat{g}^{(Y)}\in\hat{G}^{(Y)}}\left|\frac{1}{d^{++}}\sum_{i\in \hat{g}^{(X)}}\sum_{j\in \hat{g}^{(Y)}}B_{ij}\right| \label{rowMod}
\end{equation}
and
\begin{equation}
Q_{column}\left(\hat{g}^{(Y)}\right)=\sum_{\hat{g}^{(X)}\in\hat{G}^{(X)}}\left|\frac{1}{d^{++}}\sum_{i\in \hat{g}^{(X)}}\sum_{j\in \hat{g}^{(Y)}}B_{ij}\right|. \label{columnMod}
\end{equation}
\end{Definition}
\noindent Considering the absolute values of the local co-modularities in these sums serves to prioritise the most extreme choices of divisions of nodes into co-communities, according to their local co-modularities. On the other hand if absolute values were not considered here, the row and column modularities would always be zero, because the rows and columns of $\mathbf{B}$ must always sum to zero. We note that we could have chosen to use squared instead of absolute values, however we choose to use absolute values so as not to give extra weight to a few extreme values. The row and column co-modularities are the sums, respectively, of the absolute values of the local co-modularities along the rows and columns respectively, of the ordered adjacency matrix. Hence, they represent a measure of how extreme the co-community divisions are, in each row and column, according to the groupings defined by $\hat{G}^{(X)}$ and $\hat{G}^{(Y)}$. By ordering the $X$-node and $Y$-node groupings by decreasing $Q_{row}\left(\hat{g}^{(X)}\right)$ and $Q_{column}\left(\hat{g}^{(Y)}\right)$ respectively, co-communities with the largest local co-modularities will tend to congregate towards the top-left of the ordered adjacency matrix. This is a natural arrangement for visualisation as a heatmap, because it tends to place the strongest co-communities together in this corner, and so the attention is intuitively drawn to this region. 

We note that there may be other equally effective ways of arranging the adjacency matrix for visualisation as a heatmap. However, this method is effective, and it is a parsimonious solution in the context of co-modularity, because row and column modularities are very simply and intuitively related to local co-modularity. In the case that there is no co-community structure present, such as under the null model of equation \ref{DCSCBMnull}, then $Q_{row}$ and $Q_{column}$ as defined in Definition \ref{rcCoModDef} would also tend to be close to zero, and the ordering would cease to be meaningful. However, if there are even a few significant co-communities present, their corresponding $X$ and $Y$-node groupings $\hat{g}^{(X)}$ and $\hat{g}^{(Y)}$ would stand out, as assessed by $Q_{row}$ and $Q_{column}$. Therefore these $\hat{g}^{(X)}$ and $\hat{g}^{(Y)}$ would be placed at the top of the respective orderings, with the co-community pairings tending towards in the top-left corner. The other rows and columns, which do not contain significant co-communities, would have corresponding $Q_{row}$ and $Q_{column}$ close to zero. Hence, these  rows and columns would be naturally ordered according to their irrelevance. They would accordingly be placed further away from the top-left of the heatmap, giving the intuition that they are unimportant.

\subsection{Defining an objective function for optimising the co-community partitions}\label{DefObjFunOptCoCom}
Defining an objective function over the whole network, in terms of the  assignments of the nodes to $X$-node and $Y$-node groupings $\hat{g}^{(X)}$ and $\hat{g}^{(Y)}$, allows optimisation of these node assignments. It also provides a means of comparison of algorithmic parameters and other design choices in the practical implementation of the methods. It would be most ideal, for a trial assignment of nodes to $\hat{G}^{(X)}$ and $\hat{G}^{(Y)}$, to estimate the set of co-communities $\hat{C}$ using the method of Section \ref{compAssessSigCoCom}, and then to calculate the co-modularity according to Definition \ref{coModDef}. However, for a large number of repetitions within an algorithm, or for an iterative search and optimisation, this would be computationally inefficient. Instead, we define the global co-modularity to be used as an objective function for such purposes, as follows:
\begin{Definition}\label{globalCoModDef}
With $d^{++}$ given by Definition \ref{coModDef}, and with $\mathbf{B}$ given by Definition \ref{localCoModDef}, with with the set of $X$-node groupings and the set of $Y$-node groupings estimated according to Algorithm \ref{specClustAlg} as $\hat{G}^{(X)}$ and $\hat{G}^{(Y)}$ respectively, the global co-modularity is defined, for $\hat{g}^{(X)}\in\hat{G}^{(X)}$ and $\hat{g}^{(Y)}\in\hat{G}^{(Y)}$, as:
\begin{equation}
Q_{global}=\sum_{\hat{g}^{(Y)}\in\hat{G}^{(Y)}}\sum_{\hat{g}^{(X)}\in\hat{G}^{(X)}}\left|\frac{1}{d^{++}}\sum_{i\in \hat{g}^{(X)}}\sum_{j\in \hat{g}^{(Y)}}B_{ij}\right|. \label{globalMod}
\end{equation}
\end{Definition}
\noindent For a pairing $\hat{g}^{(X)}$ and $\hat{g}^{(Y)}$, the local co-modularity $Q_{XY}\left(\hat{g}^{(X)},\hat{g}^{(Y)}\right)$ represents the strength of the co-community structure in that grouping of $X$-nodes and $Y$-nodes. If the absolute value was not considered in the sum, $Q_{global}$ would always be zero. Hence, by prioritising a sum of the absolute values of the local co-modularity of all pairings $\hat{g}^{(X)}$ with $\hat{g}^{(Y)}$, we prioritise an extreme division of the $X$-nodes and $Y$-nodes into co-communities, as measured by the local co-modularity. This therefore corresponds to an extreme partition in terms of co-community structure, as assessed by co-modularity. 

Spectral clustering usually requires the nodes to be grouped in the spaces of the top singular vectors of the co-Laplacian, and this grouping is often carried out by $k$-means, as described in Algorithm \ref{specClustAlg}. Because $k$-means optimisation is not convex, the converged result may be a local optimum. Hence, implementations of $k$-means often begin at a random start-point, with the optimisation run several times from random start-points, choosing the result which is in some sense optimal. In the community-detection setting, a natural statistic to maximise in this optimisation is the Newman-Girvan modularity. An equivalent statistic here to maximise in this co-community detection setting is hence the global co-modularity, which is intuitively linked to the local co-modularity measure of co-community structure. In the community-detection setting, assignments to communities can also be optimised by carrying out node-swapping between communities, in order to maximise the Newman-Girvan modularity \citep{blondel2008fast}. The global co-modularity is a statistic which could be equivalently maximised, in this co-community detection setting.

\section{Examples}\label{examples}
In this section, we present results of applying the proposed methodology to simulated data, and to real data relating to movie reviews and to genomic patterns. We fit the degree-corrected stochastic co-blockmodel by spectral clustering as described in Section \ref{coModCoComDet}, with additional practical details as described below. As noted in Section \ref{introSect}, methodology to determine the optimal number of clusters is a challenging problem, worth studying independently of methodology to determine the clusters themselves \citep{tibshirani2001estimating}. Therefore to enable us to properly evaluate our proposed methodology for detecting co-communities (as illustrated by the shaded blocks in Fig.\ref{coComBaseEg}), we use the ground-truth cluster numbers defined in the simulation study to focus on testing whether we can recover the planted co-communities. Then in Section \ref{movieSect}, we are limited practically in how many clusters we should aim to detect by the granularity of the ground-truth (i.e., node covariate information) that is available. Again, our primary aim is to assess the proposed methodology for detecting co-communities: in this example, a co-community represents a group of similar movie fans who like a set of similar movies. We also note that Appendix C provides a practical method with theoretical justification for estimating the optimal number of $X$ and $Y$-node groupings $k^{(X)}$ and $k^{(Y)}$, from which co-communities can be identified. 

\begin{wrapfigure}{r}{0.5\textwidth}
\vspace{-2ex}
\begin{center}
\includegraphics[width=0.48\textwidth]{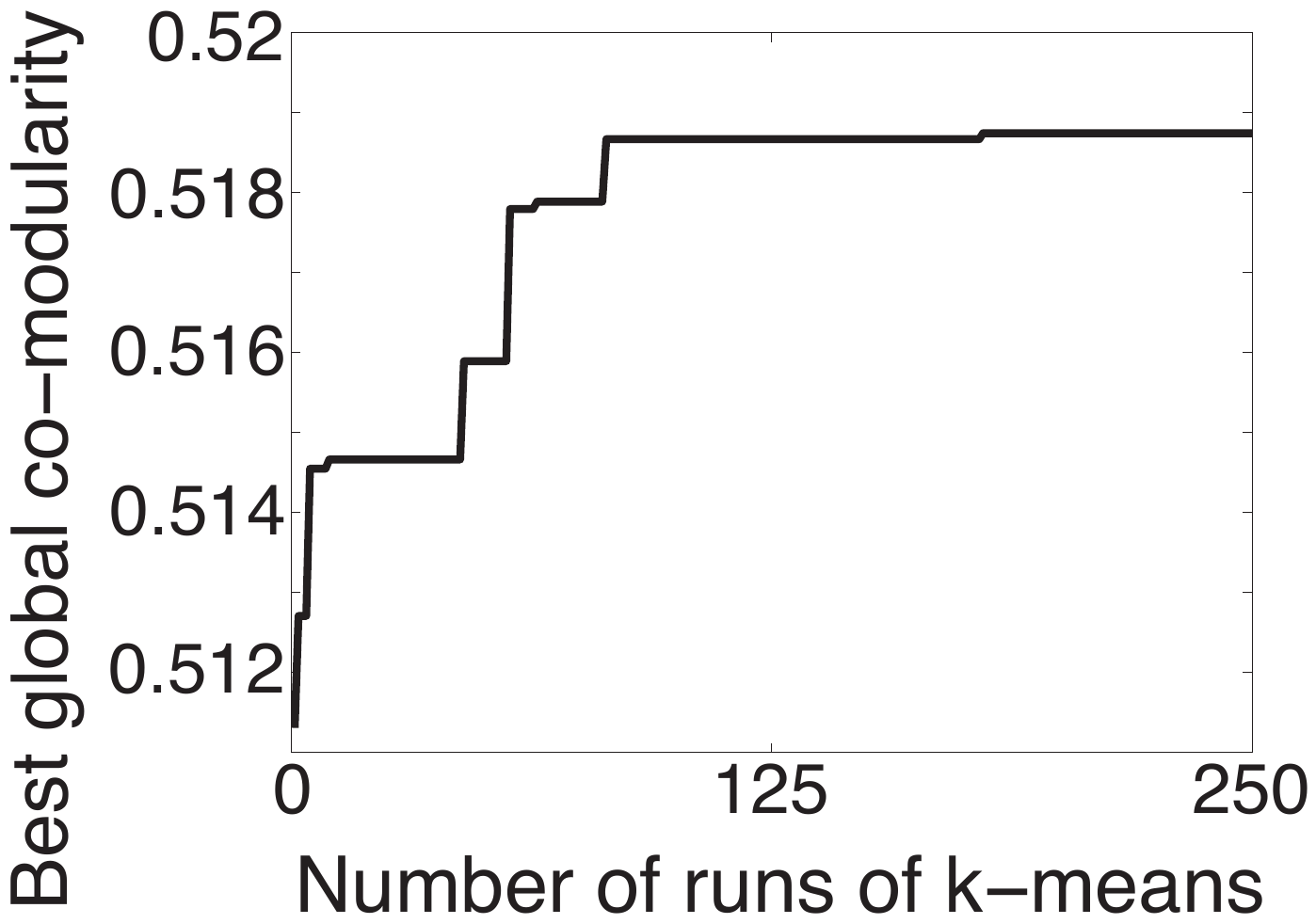}
\end{center}
\vspace{-2ex}
\caption{Convergence of the co-modularity.} \label{convergence}
\caption*{The co-modularity converges well to a maximum, within 250 runs of $k$-means, in the genomics data set. For reference, the co-modularity is consistently found to be 0 when calculated based on randomly assigned co-community partitions of similar size.}
\vspace{2ex}
\end{wrapfigure}
In the context of community detection, fitting the degree-corrected stochastic blockmodel using spectral clustering, when calculating the Laplacian it is advantageous to slightly inflate the degree distribution (regularisation) \citep{qin2013regularized}, a trick which made Google's original page-rank algorithm \citep{page1999pagerank} so effective in web-searching. Here in the co-community detection setting, correspondingly when calculating the co-Laplacian (equation \ref{coLaplacian}), we inflate the diagonals of $\textbf{D}^{(X)}$ and $\textbf{D}^{(Y)}$ by the medians of $\mathbf{d}^{(X)}$ and $\mathbf{d}^{(Y)}$, respectively. Further, when fitting variants of the stochatic blockmodel by spectral clustering with $k$-means, nodes with small leverage score (which are usually low-degree nodes) can be excluded from the $k$-means step \citep{qin2013regularized}; this practice is also followed here. We note that these regularisation steps have not previously been carried out in this co-community detection / co-clustering setting. We also note that while spectral clustering is in general computationally intensive, binary adjacency matrices such as those dealt with in this setting tend to be very sparse. Further, we only require $k=\mathrm{Max}\left(k^{(X)},k^{(Y)}\right)$ components of the singular value decomposition, a number which tends to be two or more orders of magnitude smaller than the maximum dimension of the adjacency matrix. Efficient computational methods exist to find the top few components in the singular value decomposition of large sparse matrices \citep{sorensen1992implicit,lehoucq1996deflation}, with implementations in \textit{Matlab} and \textit{R}, meaning that these methods are easy to implement and practical for large networks. The $k$-means clustering algorithm begins with a random start-point, and hence it can provide a different result each time it is run. We therefore run the $k$-means step in the spectral clustering several times, choosing the result which maximises the global co-modularity (equation \ref{globalMod}). We run $k$-means repeatedly until the output is visually assessed to have stabilised, at which point it can be seen from the convergence plot that there is very little, if any, improvement in co-modularity achieved by further runs of $k$-means. An example of such convergence in the genomics data set presented in Section \ref{IGVanalysisSection} is shown in Figure \ref{convergence}.

\subsection{Simulation Study}
We carried out a simulation study, to evaluate the effectiveness of this co-community detection methodology against generated networks with known ground-truth co-communities. A classic generative model for exchangeable random networks with heterogenous degrees is the logistic-linear model \citep{perry2012null}. We use a version here for bipartite networks, with additional co-community structure, defined as:
\begin{equation}
\mathrm{Logit}\left(p_{ij}\right)=\alpha^{(X)}_i+\alpha^{(Y)}_j+\theta_{ij},\label{genModelEq}
\end{equation}
where $p_{ij}$ defines the probability of an edge being observed between nodes $i$ and $j$. We choose to use this model, because the parameters can take any real values, and the edge probabilities $p_{ij}$ will still be between 0 and 1. This model only deviates from the equivalent log model when the parameter values become very large, which is what prevents $p_{ij}$ from reaching (and exceeding) 1 \citep{perry2012null}. Further, the blockmodel approximates any smooth function, and hence the model can be used purely in the sense of approximation \citep{olhede2014net,choi2014co}. The node-specific parameters $\alpha^{(X)}_i$ and $\alpha^{(Y)}_j$ are elements of parameter vectors $\boldsymbol{\alpha}^{(X)}$ and $\boldsymbol{\alpha}^{(Y)}$ which define degree distributions for the $X$ and $Y$-nodes. We choose power-law degree distributions for the nodes, because this is a characteristic of scale-free networks \citep{barabasi2004network}, which are found to be physically realistic in a wide range of scenarios, including biological networks \citep{wagner2002estimating}, and social networks \citep{barabasi1999emergence}. We note that although power-law degree distributions are not found in exchangeable networks, blockmodels (such as the model presented here) are still good at approximating the propensity for connections within the network \citep{borgs2015consistent}. The parameters $\alpha^{(X)}_i$ and $\alpha^{(Y)}_j$ are each generated as the logarithms of samples taken from a bounded Pareto distribution as in \citep{olhede2012degree}. We note that because $\alpha^{(X)}_i$ and $\alpha^{(Y)}_j$ are chosen to be random, our generated networks are exchangeable \citep{kallenberg2005probabilistic}, whereas if $\alpha^{(X)}_i$ and $\alpha^{(Y)}_j$ were defined deterministically, these networks would instead be generated under the inhomogenous random graph model \citep{bollobas2007phase,soderberg2002general}. The co-community parameter $\theta_{ij}$ is allowed to take two values: $\theta_{ij}=\theta_{\mathrm{in}}$ if $i$ and $j$ are in the same co-community, and $\theta_{ij}=\theta_{\mathrm{out}}$ otherwise, which is equivalent to the modelling constraint we applied in equation \ref{poissonCons}. After generating the $p_{ij}$ according to \ref{genModelEq}, the network is generated by sampling each $A_{ij}$,
\begin{equation*}
A_{ij}\sim\mathrm{Bernouilli}\left(p_{ij}\right).
\end{equation*}

The co-communities themselves are planted in the network as randomly chosen groups of 150 of each type of node. We note that some nodes may not lie in any co-community: for nodes not in a co-community, the edge probability is regulated by $\theta_{\mathrm{out}}$, and similarly by $\theta_{\mathrm{in}}$ for two nodes that are in a co-community. The maximum number of co-communities is set at $k^{(X)}\times k^{(Y)}$, and by analogy with the unipartite/symmetric community detection setting, we choose to set the number of co-communities $T$ as the square-root of this theoretical maximum, $T=\sqrt{k^{(X)}\times k^{(Y)}}$. As discussed in Section \ref{idenCompCoCom}, in the unipartite community detection setting there is a constraint on the number of communities, $k=k^{(X)}=k^{(Y)}$, because the $X$-node and $Y$-node groupings are the same thing. This constraint does not exist in the bipartite co-community detection setting, and so the theoretical maximum number of co-communities is $k^{(X)}\times k^{(Y)}$, i.e., the square of the number of communities in the equivalent symmetric community detection setting. However, we expect the number of co-communities to be significantly less than this in practice, and so by default, we choose $T=\sqrt{k^{(X)}\times k^{(Y)}}$ as the number of co-communities, although we note that many other choices would also be valid here. 

We test the methods on networks generated with $k^{(X)}$ and $k^{(Y)}$ ranging from 8 and 6 respectively up to 80 and 60 respectively (corresponding to values of numbers of nodes, $m$ and $l$, ranging from 1200 and 900 up to 12000 and 9000, respectively). We also test the methods on networks generated with values of $\theta_{\mathrm{in}}$ from 10 to 50, which corresponds to within co-community edge density $\rho_\mathrm{in}\in\{0.039,0.15,0.34,0.6\}$, and we set $\theta_{\mathrm{out}}=1$, corresponding to outside or between co-community edge density $\rho_\mathrm{in}=0.0013$ (N.B., $\theta_{\mathrm{in}}$ and $\theta_{\mathrm{out}}$ are not probabilities, equation \ref{genModelEq}). For each combination of parameters, we carry out 50 repetitions of network generation and co-community detection, to enable assessment of the variability of the accuracy of the co-community detection (with more repetitions, the computational cost becomes prohibitive). After generating the networks, we detect co-communities according to the methods described above, based on the same values of $k^{(X)}$ and $k^{(Y)}$ that we used to generate the networks. We keep these values the same, to understand specifically how the co-community detection methodology is working, as discussed in more detail in Section \ref{introSect}. This means there are $k^{(X)}\times k^{(Y)}$ potential co-communities, and we assess each in terms of strength and significance, as discussed in Section \ref{compAssessSigCoCom}. Hence, we define the estimated set of co-communities $\hat{C}$, as all combinations of detected $X$ and $Y$-node groupings $\hat{g}^{(X)}\in\hat{G}^{(X)}$ with $\hat{g}^{(Y)}\in\hat{G}^{(Y)}$ which are significant according to a $z$-test with zero mean and variance calculated as in Equation \ref{coModVarEst}. We define significance according to FDR (false discovery rate) corrected \citep{benjamini1995controlling} $p\mathrm{-value}<0.05$. This tends to result in more co-communities being detected than were originally planted (primarily due to some being split), however we note that the main aim of this methodology is to find a good representation of the underlying co-community structure (as assessed by co-modularity), rather than to reproduce it exactly.

To compare detected co-communities with the ground-truth planted co-communities, we use the normalised mutual information (NMI) \citep{danon2005comparing} to compare the corresponding $X$ and $Y$-node groupings (over the full node-sets). The NMI compares the numbers of nodes which appear together in the discovered $X$ and $Y$-node groupings, compared with whether they appeared together in those that correspond to the planted co-communities (adjusted for group sizes). NMI has been used in a similar way previously in the co-community detection context \citep{larremore2014efficiently}, as well as the unipartite community-detection context \citep{zhao2012consistency}. The NMI takes the value 1 if the  $X$ and $Y$-node groupings that correspond to these co-communities are perfectly reproduced in the co-community detection, and 0 if they are not reproduced at all, and somewhere in between if they are partially reproduced. The results, together with examples of randomly generated adjacency matrices, are shown in Figure \ref{simRes}, which shows that the method performs well as long as there is sufficient within-co-community edge density (implying a detection threshold), and performs well as the number of co-communities increases. The run-times for the fits for number of row nodes equal to 1200, 4800, 8400 and 12000 (with number of column nodes equal to 900, 3600, 8400 and 9000) are 1.6s, 21s, 80s and 190s, respectively (using one core of a MacBook Pro, 2019, 2.6 GHz).  
\begin{figure}[ht!]
\vspace{-1ex}
\includegraphics[width=1\textwidth]{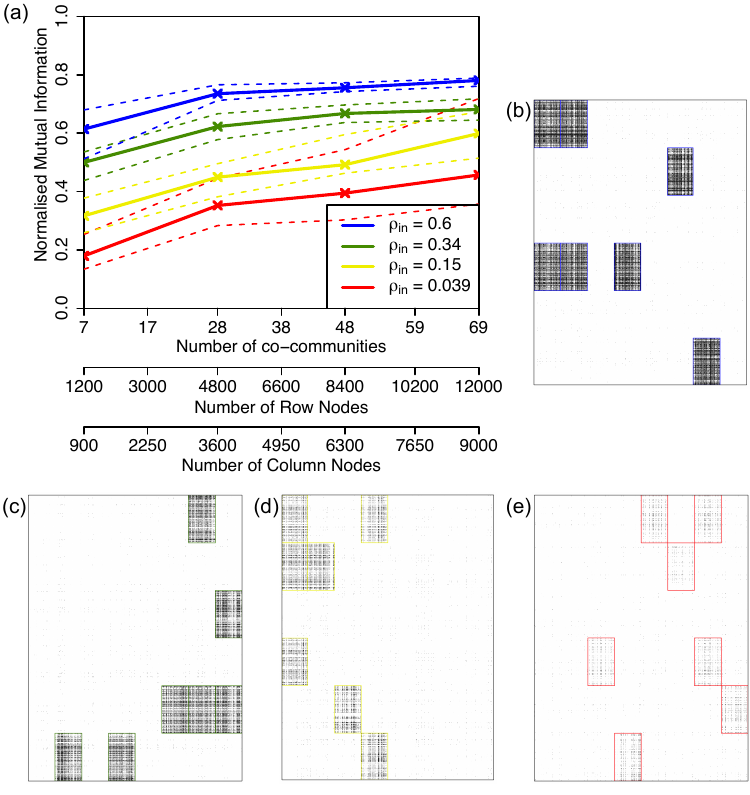}
\caption{Simulation Study.} \label{simRes}
\caption*{(a) Normalised mutual information (NMI) compares detected co-communities with ground-truth planted co-communities (dashed lines indicate 95\% C.I.) (b)-(e) Examples of generated networks all with $nR=1200$, $nC=900$, $kR=8$, $kC=6$, and 7 planted co-communities; entries in the adjacency matrix equal to 1 (representing a network edge) are marked in black; planted co-communities are outlined in colour. (b) $\theta_{\mathrm{in}}=40$, within-community edge density $\rho_\mathrm{in}=0.6$; (b) $\theta_{\mathrm{in}}=30$, $\rho_\mathrm{in}=0.34$; (c) $\theta_{\mathrm{in}}=20$, $\rho_\mathrm{in}=0.15$; (d) $\theta_{\mathrm{in}}=10$, $\rho_\mathrm{in}=0.039$. For all networks, $\theta_{\mathrm{out}}=1$, outside/between co-community edge density $\rho_\mathrm{out}=0.0013$}
\end{figure}

\vspace{-1ex}
\subsection{Genomics Data-set}\label{IGVanalysisSection}
We present an example of a practical application of these methods to a challenging problem in genomics. A gene encodes how to make a gene product, and the corresponding gene expression level quantifies how much gene product is currently being produced. Hence, the gene expression level indicates the extent to which a gene is `active', or `switched on'. DNA methylation is a gene regulatory pattern, meaning that it influences the activity/expression level of particular genes. DNA methylation patterns are themselves influenced by the expression levels/activity of other genes. However, much is still unknown about the interaction between DNA methylation patterns and gene expression patterns \citep{jones2012functions}. It is of much interest to uncover groups of genes with methylation patterns which are linked to the expression patterns of other groups of genes, to allow biological hypotheses to be formed, which can then be investigated further, experimentally and computationally. Hence, this is a natural scenario to be approached with co-community detection, as the method offers the potential to uncover latent structure not easily identifiable otherwise. 

\begin{figure}[!ht]
\includegraphics[width=0.95\textwidth]{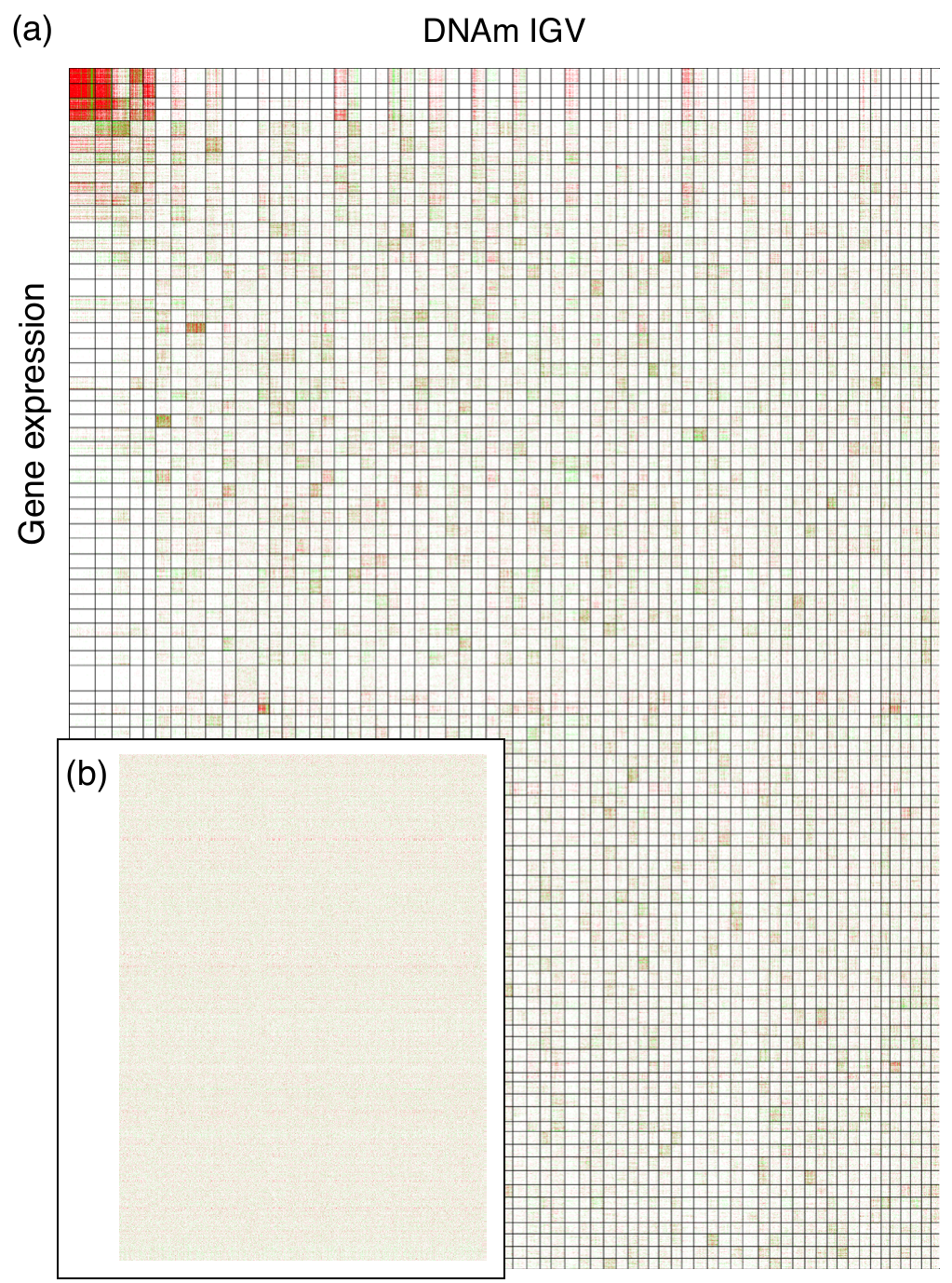}
\caption{Co-communities in the genomics data set.} \label{exprCorIGV}
\caption*{(a) Genes are ordered along the margins of the adjacency matrix, according to co-communities detected by the methods presented here. Partitions between detected co-communities are shown with black lines. (b) The same adjacency matrix ordered along its margins alphabetically by gene name, i.e., without ordering the margins using co-community detection. Entries in the adjacency matrix equal to 1 (representing a network edge) are coloured, with green and red indicating positive and negative associations, respectively. We note that this is a signed network \citep{leskovec2010signed}. N.B., this colour-scheme contrasts with that used in Figure \ref{movieLens}, in which network edges do not have associated signs, and hence are all coloured blue.}
\vspace{-5ex}
\end{figure}

As a measure of the DNA methylation (DNAm) pattern of each gene, we choose to consider here intra-gene DNA methylation variability (IGV), as it is a per-gene measure of DNA methylation variance which has been shown to be strongly associated with disease \citep{bartlett2013corruption,bartlett2015intra}. We denote the gene expression variables $X(i)$, $i=1,...,m$ and the DNAm variables $Y(j)$, $j=1,...,l$; i.e., $X(i)$ and $Y(j)$ refer to the measurements for particular genes of gene expression and DNA methylation IGV respectively. We define a network edge $A_{ij}=1$ if variables $X(i)$ and $Y(i)$ are significantly correlated, and we set $A_{ij}=0$ otherwise \citep{bartlett2016network}. We carried out co-community detection on this genomics data set according to the methods described above (data source: The Cancer Genome Atlas \citep{hampton2006cancer}, breast cancer invasive carcinoma data set, basal tumour samples only). Figure \ref{exprCorIGV}(a) shows the adjacency matrix after carrying out co-community detection, ordering the $X$ and $Y$-node groupings by row and column co-modularity (equations \ref{rowMod} and \ref{columnMod}). Figure \ref{exprCorIGV}(b) (inlay) shows the same adjacency matrix ordered along its margins alphabetically by gene name, i.e., without ordering the margins using co-community detection. Hence, Figure \ref{exprCorIGV}(b) shows a baseline in which the nodes are essentially randomly ordered, against which to compare the adjacency matrix after co-community detection, and ordering based upon it. The co-community structure is clearly revealed in Figure \ref{exprCorIGV}(a), whereas no co-community structure is visible in Figure \ref{exprCorIGV}(b). We define a co-community $\hat{c}\in\hat{C}$ as a combination of $X$-node grouping  $\hat{g}^{(X)}\in\hat{G}^{(X)}$ with $Y$-node grouping  $\hat{g}^{(Y)}\in\hat{G}^{(Y)}$ which is significant according to a $z$-test with zero mean and variance calculated as in equation \ref{coModVarEst}, with significance defined by FDR-corrected $p\mathrm{-value}<0.05$. The numbers of $X$ and $Y$-node groupings, $k^{(X)}$ and $k^{(Y)}$, are estimated according to equations \ref{estKx} and \ref{estKy} as 89 and 67 respectively, leading to 5963 potential co-communities, of which $\hat{T}=2018$ are found to be significant. We tested these 2018 significant co-communities for domain relevance, by comparing the overlap of the genes (nodes) of each co-community, separately with each of $10295$ known gene-groups (data source: \allowbreak\textit{http://www.broadinstitute.org/gsea/msigdb/}). This type of analysis is often called `gene set enrichment analysis' (GSEA) \citep{subramanian2005gene}. We found that 1340 (66\%) overlap significantly (Fisher's exact test, FDR-adjusted $p<0.05$) with these known gene-sets, confirming the domain relevance of this result, as well as indicating novel findings which could be investigated further by experimental biologists. 

\subsection{Movie-review Data-set}\label{movieSect}
We present a second, contrasting example of a practical application of these methods to real data, to a consumer-product review dataset. We downloaded movie review data from the \textit{Movie Lens} database, which details $1\ 000\ 209$ reviews of 3952 different movies, by 6040 unique users who each provided at least 20 different reviews (data source: \textit{http://grouplens.org/datasets/movielens/}). Denoting movies by the variables $X(i)$, $i=1,...,m$ and users by the variables $Y(j)$, $j=1,...,l$, we define a network edge, i.e., $A_{ij}=1$, if movie $X(i)$ has been reviewed by user $Y(i)$, and no edge, i.e., $A_{ij}=0$, otherwise. Covariate information is also available, assigning each movie to one of 18 categories, and classifying each user into one of 7 age groups and 20 professions.

As discussed in Section \ref{introSect}, determining the optimal number of clusters is a different problem than determining the clusters themselves. In this example, the granularity of the available ground-truth clusters (i.e., the covariate information we have available for verification of detected co-clusters) is very much less than that estimated according to equations \ref{estKx} and \ref{estKy} (i.e., 77 and 95 respectively). To make the comparison with the available ground-truth straightforward, we choose to use a level of granularity of $k^{(X)}=10$ and $k^{(Y)}=15$ that is in line with the available ground-truth. This is well justified theoretically, as follows. The graphon function \citep{wolfe2013nonparametric} regulates the smoothness of the success probabilities generating the edges of the network, and so if the success probabilities are changing rapidly across nodes, then we need to use more blocks, or communities - in doing so, we are ensuring that even the roughest portion of the graphon function are sufficiently well resolved. However, we can also reduce granularity by reducing the number of blocks, whilst accepting that we will lose the ability to resolve some portions of the graphon function well. This point is also discussed further in Appendix C.

\begin{figure}[!ht]
\includegraphics[width=0.95\textwidth]{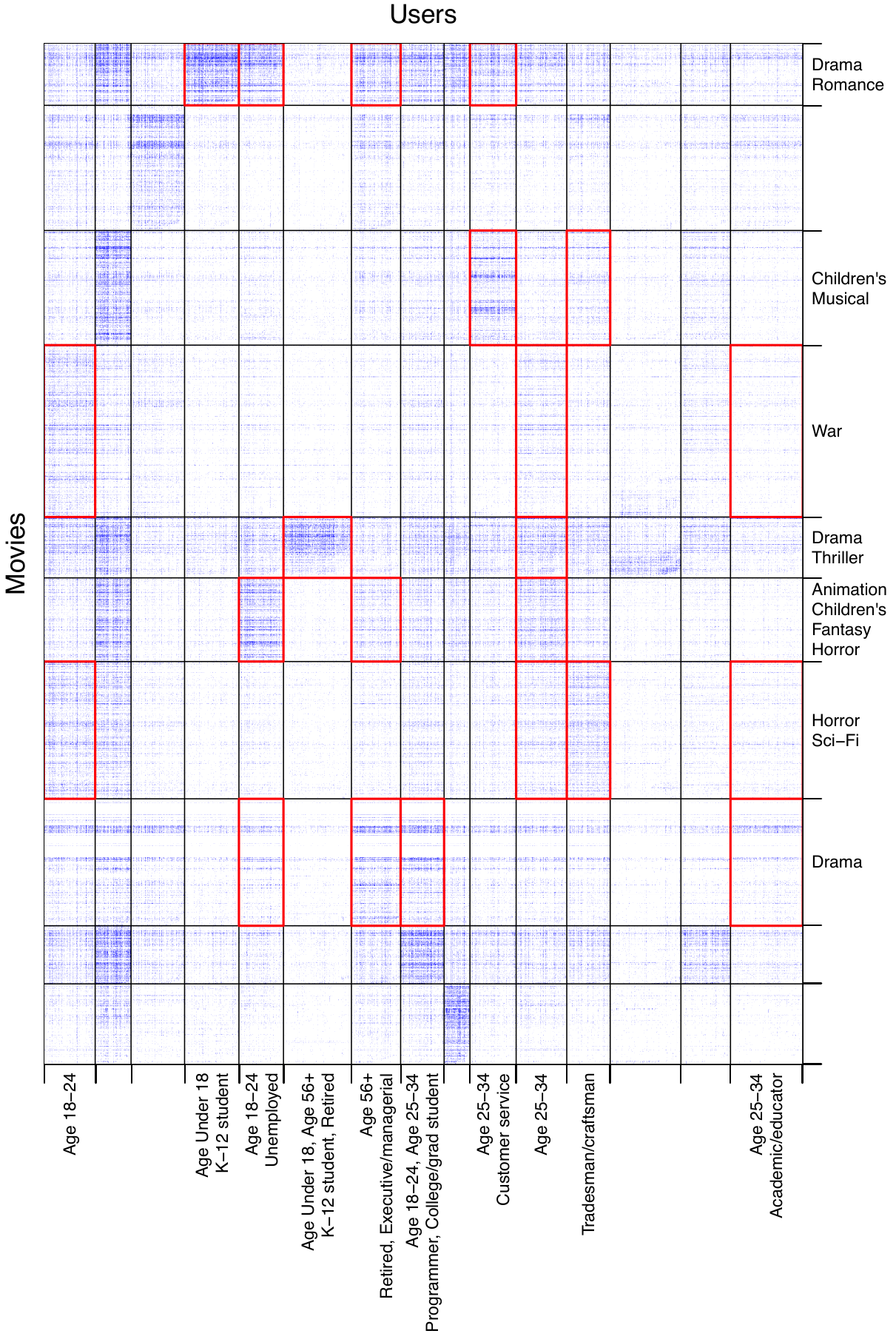}
\caption{Co-communities in the movie-review data set.} \label{movieLens}
\caption*{Entries in the adjacency matrix equal to 1 (representing a network edge) are coloured blue, and detected communities are outlined in black.}
\vspace{-8ex}
\end{figure}

Co-community detection was carried out using the methods described above, and Figure \ref{movieLens} shows the result. Of the 150 potential co-communities, $\hat{T}=41$ are found to be significant using the $z$-test as before, and are thus defined as co-communities. The $X$ and $Y$-nodes of these 41 co-communities are tested for overlap with the covariate groups. Of these, 22 are found to overlap significantly with one or more covariate groups, and these are highlighted in red, with the significant covariate groups shown along the margins, in Figure \ref{movieLens}. We use Fisher's exact test to assess significances, defined as FDR-corrected $p<0.05$. We note that using the Fisher test is slightly inappropriate due to the generative mechanism of the groups. However we would anticipate that any residual correlation from the group discovery is minor. Importantly though for assessment via this benchmark data-set, many of the findings are predictable: horror, sci-fi and war films tend to be watched by younger people; drama and romance are popular across the board. Others need more explanation, for example a group of children's movies and musicals tend to be reviewed by 25-34 year old customer service professionals. However, we can expect that this is a demographic group of people who tend to have younger children who they watch movies with. Other children's movies are grouped together with animation, fantasy and horror, and tend to be watched by both younger and older groups. This might reflect very broad classifications used for such movies, many of which in reality could be fairly similar. Also that these are groups of people who would tend to watch movies together. An important conclusion to draw, is that the covariate information available for this data-set appears to be of a lower granularity than the detail which can be revealed by these co-community detection methods.

\section{Conclusion}
We have introduced the notion of co-modularity based on the stochastic co-blockmodel, and have shown how it can be used to perform co-community detection in bipartite networks. We have shown how co-modularity can be used to compare co-communities, to calculate their strength and significance, and to arrange them for visualisation. We have addressed practical points about the implementation of the methodology, and have demonstrated its usefulness with a simulation study and application to two contrasting examples of real datasets, from genomics and consumer-product reviews. We note that the main aim of our method is detection of co-communities (in fact mis-specified, because we don't think that the inferred groups are perfect). An interesting extension to this methodology would be to consider overlapping blocks in the stochastic co-blockmodel, a problem which has already been successfully addressed in the context of the stochastic blockmodel for unipartite networks \citep{latouche2011overlapping}, and in co-clustering without fitting the stochastic blockmodel \citep{madeira2004biclustering}. Another interesting application would be to develop an online version of the method (i.e., which updates rather than re-computes) as a computationally efficient approach to large and growing data-sets \citep{zanghi2010strategies}. This methodology would be expected to work similarly well in many other contexts, such as interpersonal networks where the individuals are of two distinct categories, such as teachers and students, or publication networks where the two types of variables are authors and papers. This methodology could also be expected to work in even more general settings of bi-clustering or co-clustering, in which the variables being clustered together are simply correlated, rather than having any tangible interactive behaviour in the real world. These methods are based on commonly available computationally efficient methods for large sparse matrices, and perform well on large datasets, with large numbers of co-communities, often performing better than methods based on model likelihoods.

\bibliography{references}

\begin{thebibliography}{}

\bibitem[\protect\citeauthoryear{Airoldi, Blei, Fienberg, and Xing}{Airoldi
  et~al.}{2009}]{airoldi2009mixed}
Airoldi, E.~M., D.~M. Blei, S.~E. Fienberg, and E.~P. Xing 2009.
\newblock Mixed membership stochastic blockmodels.
\newblock In {\em Advances in Neural Information Processing Systems}, pp.\
  33--40.

\bibitem[\protect\citeauthoryear{Aldous}{Aldous}{1985}]{aldous1985exchangeability}
Aldous, D.~J. 1985.
\newblock {\em Exchangeability and related topics}.
\newblock Springer.

\bibitem[\protect\citeauthoryear{Amini, Chen, Bickel, and Levina}{Amini
  et~al.}{2013}]{amini2013pseudo}
Amini, A.~A., A.~Chen, P.~J. Bickel, and E.~Levina 2013.
\newblock Pseudo-likelihood methods for community detection in large sparse
  networks.
\newblock {\em The Annals of Statistics\/}~{\em 41\/}(4), 2097--2122.

\bibitem[\protect\citeauthoryear{Barab{\'a}si and Albert}{Barab{\'a}si and
  Albert}{1999}]{barabasi1999emergence}
Barab{\'a}si, A.-L. and R.~Albert 1999.
\newblock Emergence of scaling in random networks.
\newblock {\em Science\/}~{\em 286\/}(5439), 509--512.

\bibitem[\protect\citeauthoryear{Barab{\'a}si and Oltvai}{Barab{\'a}si and
  Oltvai}{2004}]{barabasi2004network}
Barab{\'a}si, A.-L. and Z.~N. Oltvai 2004.
\newblock Network biology: understanding the cell's functional organization.
\newblock {\em Nature Reviews Genetics\/}~{\em 5\/}(2), 101--113.

\bibitem[\protect\citeauthoryear{Bartlett}{Bartlett}{2017}]{bartlett2016network}
Bartlett, T.~E. 2017.
\newblock Network inference and community detection, based on covariance
  matrices, correlations and test statistics from arbitrary distributions.
\newblock {\em Communications in Statistics - Theory and Methods\/}~{\em
  46\/}(18), 9150--9165.

\bibitem[\protect\citeauthoryear{Bartlett, Jones, Goode, Fridley, Cunningham,
  Berns, Wik, Salvesen, Davidson, Trope, et~al.}{Bartlett
  et~al.}{2015}]{bartlett2015intra}
Bartlett, T.~E., A.~Jones, E.~L. Goode, B.~L. Fridley, J.~M. Cunningham, E.~M.
  Berns, E.~Wik, H.~B. Salvesen, B.~Davidson, C.~G. Trope, et~al. 2015.
\newblock Intra-gene dna methylation variability is a clinically independent
  prognostic marker in women?s cancers.
\newblock {\em PloS one\/}~{\em 10\/}(12), e0143178.

\bibitem[\protect\citeauthoryear{Bartlett, Zaikin, Olhede, West, Teschendorff,
  and Widschwendter}{Bartlett et~al.}{2013}]{bartlett2013corruption}
Bartlett, T.~E., A.~Zaikin, S.~C. Olhede, J.~West, A.~E. Teschendorff, and
  M.~Widschwendter 2013.
\newblock Corruption of the intra-gene {DNA} methylation architecture is a
  hallmark of cancer.
\newblock {\em PloS One\/}~{\em 8\/}(7), e68285.

\bibitem[\protect\citeauthoryear{Bassett, Porter, Wymbs, Grafton, Carlson, and
  Mucha}{Bassett et~al.}{2013}]{bassett2013robust}
Bassett, D.~S., M.~A. Porter, N.~F. Wymbs, S.~T. Grafton, J.~M. Carlson, and
  P.~J. Mucha 2013.
\newblock Robust detection of dynamic community structure in networks.
\newblock {\em Chaos: An Interdisciplinary Journal of Nonlinear Science\/}~{\em
  23\/}(1), 013142.

\bibitem[\protect\citeauthoryear{Benjamini and Hochberg}{Benjamini and
  Hochberg}{1995}]{benjamini1995controlling}
Benjamini, Y. and Y.~Hochberg 1995.
\newblock Controlling the false discovery rate: A practical and powerful
  approach to multiple testing.
\newblock {\em Journal of the Royal Statistical Society. Series B
  (Methodological)\/}~{\em 57\/}(1), 289--300.

\bibitem[\protect\citeauthoryear{Bhattacharya and Cui}{Bhattacharya and
  Cui}{2017}]{bhattacharya2017gpu}
Bhattacharya, A. and Y.~Cui 2017.
\newblock A gpu-accelerated algorithm for biclustering analysis and detection
  of condition-dependent coexpression network modules.
\newblock {\em Scientific Reports\/}~{\em 7\/}(1), 4162.

\bibitem[\protect\citeauthoryear{Bickel, Diggle, Fienberg, Gather, Olkin, and
  Zeger}{Bickel et~al.}{2009}]{bickel2009springer}
Bickel, P., P.~Diggle, S.~Fienberg, U.~Gather, I.~Olkin, and S.~Zeger 2009.
\newblock {\em Springer series in statistics}.
\newblock Springer.

\bibitem[\protect\citeauthoryear{Bickel and Chen}{Bickel and
  Chen}{2009}]{bickel2009nonparametric}
Bickel, P.~J. and A.~Chen 2009.
\newblock A nonparametric view of network models and newman--girvan and other
  modularities.
\newblock {\em Proceedings of the National Academy of Sciences\/}~{\em
  106\/}(50), 21068--21073.

\bibitem[\protect\citeauthoryear{Blondel, Guillaume, Lambiotte, and
  Lefebvre}{Blondel et~al.}{2008}]{blondel2008fast}
Blondel, V.~D., J.-L. Guillaume, R.~Lambiotte, and E.~Lefebvre 2008.
\newblock Fast unfolding of communities in large networks.
\newblock {\em Journal of Statistical Mechanics: Theory and Experiment\/}~{\em
  2008\/}(10), P10008.

\bibitem[\protect\citeauthoryear{Bollob{\'a}s, Janson, and
  Riordan}{Bollob{\'a}s et~al.}{2007}]{bollobas2007phase}
Bollob{\'a}s, B., S.~Janson, and O.~Riordan 2007.
\newblock The phase transition in inhomogeneous random graphs.
\newblock {\em Random Structures \& Algorithms\/}~{\em 31\/}(1), 3--122.

\bibitem[\protect\citeauthoryear{Borgs, Chayes, Cohn, and Ganguly}{Borgs
  et~al.}{2015}]{borgs2015consistent}
Borgs, C., J.~T. Chayes, H.~Cohn, and S.~Ganguly 2015.
\newblock Consistent nonparametric estimation for heavy-tailed sparse graphs.
\newblock {\em arXiv preprint arXiv:1508.06675\/}.

\bibitem[\protect\citeauthoryear{Cai, Li, et~al.}{Cai
  et~al.}{2015}]{cai2014robust}
Cai, T.~T., X.~Li, et~al. 2015.
\newblock Robust and computationally feasible community detection in the
  presence of arbitrary outlier nodes.
\newblock {\em Annals of Statistics\/}~{\em 43\/}(3), 1027--1059.

\bibitem[\protect\citeauthoryear{Choi, Wolfe, et~al.}{Choi
  et~al.}{2014}]{choi2014co}
Choi, D., P.~J. Wolfe, et~al. 2014.
\newblock Co-clustering separately exchangeable network data.
\newblock {\em The Annals of Statistics\/}~{\em 42\/}(1), 29--63.

\bibitem[\protect\citeauthoryear{Clevert, Unterthiner, Povysil, and
  Hochreiter}{Clevert et~al.}{2017}]{clevert2017rectified}
Clevert, D.-A., T.~Unterthiner, G.~Povysil, and S.~Hochreiter 2017.
\newblock Rectified factor networks for biclustering of omics data.
\newblock {\em Bioinformatics\/}~{\em 33\/}(14), i59--i66.

\bibitem[\protect\citeauthoryear{Danon, Diaz-Guilera, Duch, and Arenas}{Danon
  et~al.}{2005}]{danon2005comparing}
Danon, L., A.~Diaz-Guilera, J.~Duch, and A.~Arenas 2005.
\newblock Comparing community structure identification.
\newblock {\em Journal of Statistical Mechanics: Theory and Experiment\/}~{\em
  2005\/}(09), P09008.

\bibitem[\protect\citeauthoryear{Dhillon}{Dhillon}{2001}]{dhillon2001co}
Dhillon, I.~S. 2001.
\newblock Co-clustering documents and words using bipartite spectral graph
  partitioning.
\newblock In {\em Proceedings of the seventh ACM SIGKDD international
  conference on Knowledge discovery and data mining}, pp.\  269--274. ACM.

\bibitem[\protect\citeauthoryear{Diaconis}{Diaconis}{1977}]{diaconis1977finite}
Diaconis, P. 1977.
\newblock Finite forms of de finetti's theorem on exchangeability.
\newblock {\em Synthese\/}~{\em 36\/}(2), 271--281.

\bibitem[\protect\citeauthoryear{Flynn and Perry}{Flynn and
  Perry}{2012}]{flynn2012consistent}
Flynn, C.~J. and P.~O. Perry 2012.
\newblock Consistent biclustering.
\newblock {\em arXiv preprint arXiv:1206.6927\/}.

\bibitem[\protect\citeauthoryear{Gao, Lu, Ma, and Zhou}{Gao
  et~al.}{2016}]{gao2016optimal}
Gao, C., Y.~Lu, Z.~Ma, and H.~H. Zhou 2016.
\newblock Optimal estimation and completion of matrices with biclustering
  structures.
\newblock {\em Journal of Machine Learning Research\/}~{\em 17\/}(161), 1--29.

\bibitem[\protect\citeauthoryear{Girvan and Newman}{Girvan and
  Newman}{2002}]{girvan2002community}
Girvan, M. and M.~E. Newman 2002.
\newblock Community structure in social and biological networks.
\newblock {\em Proceedings of the National Academy of Sciences\/}~{\em
  99\/}(12), 7821--7826.

\bibitem[\protect\citeauthoryear{Gopalan and Blei}{Gopalan and
  Blei}{2013}]{gopalan2013efficient}
Gopalan, P.~K. and D.~M. Blei 2013.
\newblock Efficient discovery of overlapping communities in massive networks.
\newblock {\em Proceedings of the National Academy of Sciences\/}~{\em 110},
  14534--14539.

\bibitem[\protect\citeauthoryear{Hampton}{Hampton}{2006}]{hampton2006cancer}
Hampton, T. 2006.
\newblock Cancer genome atlas.
\newblock {\em JAMA: The Journal of the American Medical Association\/}~{\em
  296\/}(16), 1958--1958.

\bibitem[\protect\citeauthoryear{Holland, Laskey, and Leinhardt}{Holland
  et~al.}{1983}]{holland1983stochastic}
Holland, P.~W., K.~B. Laskey, and S.~Leinhardt 1983.
\newblock Stochastic blockmodels: First steps.
\newblock {\em Social Networks\/}~{\em 5\/}(2), 109--137.

\bibitem[\protect\citeauthoryear{Hom and Johnson}{Hom and
  Johnson}{1991}]{hom1991topics}
Hom, R.~A. and C.~R. Johnson 1991.
\newblock Topics in matrix analysis.
\newblock {\em Cambridge UP, New York\/}.

\bibitem[\protect\citeauthoryear{Jones}{Jones}{2012}]{jones2012functions}
Jones, P. 2012.
\newblock Functions of {DNA} methylation: islands, start sites, gene bodies and
  beyond.
\newblock {\em Nature Reviews Genetics\/}~{\em 13\/}(7), 484--492.

\bibitem[\protect\citeauthoryear{Kallenberg}{Kallenberg}{2005}]{kallenberg2005probabilistic}
Kallenberg, O. 2005.
\newblock {\em Probabilistic symmetries and invariance principles}, Volume~9.
\newblock Springer.

\bibitem[\protect\citeauthoryear{Larremore, Clauset, and Jacobs}{Larremore
  et~al.}{2014}]{larremore2014efficiently}
Larremore, D.~B., A.~Clauset, and A.~Z. Jacobs 2014.
\newblock Efficiently inferring community structure in bipartite networks.
\newblock {\em Physical Review E\/}~{\em 90\/}(1), 012805.

\bibitem[\protect\citeauthoryear{Latouche, Birmel{\'e}, Ambroise,
  et~al.}{Latouche et~al.}{2011}]{latouche2011overlapping}
Latouche, P., E.~Birmel{\'e}, C.~Ambroise, et~al. 2011.
\newblock Overlapping stochastic block models with application to the french
  political blogosphere.
\newblock {\em The Annals of Applied Statistics\/}~{\em 5\/}(1), 309--336.

\bibitem[\protect\citeauthoryear{Lehoucq and S{\o}rensen}{Lehoucq and
  S{\o}rensen}{1996}]{lehoucq1996deflation}
Lehoucq, R.~B. and D.~C. S{\o}rensen 1996.
\newblock Deflation techniques for an implicitly restarted arnoldi iteration.
\newblock {\em SIAM Journal on Matrix Analysis and Applications\/}~{\em
  17\/}(4), 789--821.

\bibitem[\protect\citeauthoryear{Leskovec, Huttenlocher, and
  Kleinberg}{Leskovec et~al.}{2010}]{leskovec2010signed}
Leskovec, J., D.~Huttenlocher, and J.~Kleinberg 2010.
\newblock Signed networks in social media.
\newblock In {\em Proceedings of the SIGCHI conference on human factors in
  computing systems}, pp.\  1361--1370. ACM.

\bibitem[\protect\citeauthoryear{Madeira and Oliveira}{Madeira and
  Oliveira}{2004}]{madeira2004biclustering}
Madeira, S.~C. and A.~L. Oliveira 2004.
\newblock Biclustering algorithms for biological data analysis: a survey.
\newblock {\em Computational Biology and Bioinformatics, IEEE/ACM Transactions
  on\/}~{\em 1\/}(1), 24--45.

\bibitem[\protect\citeauthoryear{Newman}{Newman}{2013}]{newman2013spectral}
Newman, M.~E. 2013.
\newblock Spectral methods for community detection and graph partitioning.
\newblock {\em Physical Review E\/}~{\em 88\/}(4), 042822.

\bibitem[\protect\citeauthoryear{Newman}{Newman}{2016}]{newman2016equivalence}
Newman, M.~E. 2016.
\newblock Equivalence between modularity optimization and maximum likelihood
  methods for community detection.
\newblock {\em Physical Review E\/}~{\em 94\/}(5), 052315.

\bibitem[\protect\citeauthoryear{Newman and Girvan}{Newman and
  Girvan}{2004}]{newman2004finding}
Newman, M.~E. and M.~Girvan 2004.
\newblock Finding and evaluating community structure in networks.
\newblock {\em Physical Review E\/}~{\em 69\/}(2), 026113.

\bibitem[\protect\citeauthoryear{Olhede and Wolfe}{Olhede and
  Wolfe}{2012}]{olhede2012degree}
Olhede, S.~C. and P.~J. Wolfe 2012.
\newblock Degree-based network models.
\newblock {\em arXiv preprint arXiv:1211.6537\/}.

\bibitem[\protect\citeauthoryear{Olhede and Wolfe}{Olhede and
  Wolfe}{2014}]{olhede2014net}
Olhede, S.~C. and P.~J. Wolfe 2014.
\newblock Network histograms and universality of blockmodel approximation.
\newblock {\em Proceedings of the National Academy of Sciences\/}~{\em
  111\/}(41), 14722--14727.

\bibitem[\protect\citeauthoryear{Page, Brin, Motwani, and Winograd}{Page
  et~al.}{1999}]{page1999pagerank}
Page, L., S.~Brin, R.~Motwani, and T.~Winograd 1999.
\newblock The pagerank citation ranking: Bringing order to the web.

\bibitem[\protect\citeauthoryear{Peixoto}{Peixoto}{2019}]{peixoto2019bayesian}
Peixoto, T.~P. 2019.
\newblock Bayesian stochastic blockmodeling.
\newblock {\em Advances in network clustering and blockmodeling\/}, 289--332.

\bibitem[\protect\citeauthoryear{Perry and Wolfe}{Perry and
  Wolfe}{2012}]{perry2012null}
Perry, P.~O. and P.~J. Wolfe 2012.
\newblock Null models for network data.
\newblock {\em arXiv preprint arXiv:1201.5871\/}.

\bibitem[\protect\citeauthoryear{Qin and Rohe}{Qin and
  Rohe}{2013}]{qin2013regularized}
Qin, T. and K.~Rohe 2013.
\newblock Regularized spectral clustering under the degree-corrected stochastic
  blockmodel.
\newblock In {\em Advances in Neural Information Processing Systems}, pp.\
  3120--3128.

\bibitem[\protect\citeauthoryear{Riolo, Cantwell, Reinert, and Newman}{Riolo
  et~al.}{2017}]{riolo2017efficient}
Riolo, M.~A., G.~T. Cantwell, G.~Reinert, and M.~E. Newman 2017.
\newblock Efficient method for estimating the number of communities in a
  network.
\newblock {\em Physical review e\/}~{\em 96\/}(3), 032310.

\bibitem[\protect\citeauthoryear{Riolo and Newman}{Riolo and
  Newman}{2014}]{riolo2014first}
Riolo, M.~A. and M.~Newman 2014.
\newblock First-principles multiway spectral partitioning of graphs.
\newblock {\em Journal of Complex Networks\/}~{\em 2\/}(2), 121--140.

\bibitem[\protect\citeauthoryear{Rohe, Chatterjee, Yu, et~al.}{Rohe
  et~al.}{2011}]{rohe2011spectral}
Rohe, K., S.~Chatterjee, B.~Yu, et~al. 2011.
\newblock Spectral clustering and the high-dimensional stochastic blockmodel.
\newblock {\em The Annals of Statistics\/}~{\em 39\/}(4), 1878--1915.

\bibitem[\protect\citeauthoryear{Rohe and Yu}{Rohe and Yu}{2012}]{rohe2012co}
Rohe, K. and B.~Yu 2012.
\newblock Co-clustering for directed graphs; the stochastic co-blockmodel and a
  spectral algorithm.
\newblock {\em arXiv preprint arXiv:1204.2296\/}.

\bibitem[\protect\citeauthoryear{Rugeles, Zhao, Gao, Dash, and
  Krishnaswamy}{Rugeles et~al.}{2017}]{rugeles2017biclustering}
Rugeles, D., K.~Zhao, C.~Gao, M.~Dash, and S.~Krishnaswamy 2017.
\newblock Biclustering: An application of dual topic models.
\newblock In {\em Proceedings of the 2017 SIAM International Conference on Data
  Mining}, pp.\  453--461. SIAM.

\bibitem[\protect\citeauthoryear{S{\"o}derberg}{S{\"o}derberg}{2002}]{soderberg2002general}
S{\"o}derberg, B. 2002.
\newblock General formalism for inhomogeneous random graphs.
\newblock {\em Physical review E\/}~{\em 66\/}(6), 066121.

\bibitem[\protect\citeauthoryear{S{\o}rensen}{S{\o}rensen}{1992}]{sorensen1992implicit}
S{\o}rensen, D.~C. 1992.
\newblock Implicit application of polynomial filters in ak-step arnoldi method.
\newblock {\em SIAM Journal on Matrix Analysis and Applications\/}~{\em
  13\/}(1), 357--385.

\bibitem[\protect\citeauthoryear{Subramanian, Tamayo, Mootha, Mukherjee, Ebert,
  Gillette, Paulovich, Pomeroy, Golub, Lander, et~al.}{Subramanian
  et~al.}{2005}]{subramanian2005gene}
Subramanian, A., P.~Tamayo, V.~Mootha, S.~Mukherjee, B.~Ebert, M.~Gillette,
  A.~Paulovich, S.~Pomeroy, T.~Golub, E.~Lander, et~al. 2005.
\newblock Gene set enrichment analysis: a knowledge-based approach for
  interpreting genome-wide expression profiles.
\newblock {\em Proceedings of the National Academy of Sciences of the United
  States of America\/}~{\em 102\/}(43), 15545.

\bibitem[\protect\citeauthoryear{Tibshirani, Walther, and Hastie}{Tibshirani
  et~al.}{2001}]{tibshirani2001estimating}
Tibshirani, R., G.~Walther, and T.~Hastie 2001.
\newblock Estimating the number of clusters in a data set via the gap
  statistic.
\newblock {\em Journal of the Royal Statistical Society: Series B (Statistical
  Methodology)\/}~{\em 63\/}(2), 411--423.

\bibitem[\protect\citeauthoryear{Wagner}{Wagner}{2002}]{wagner2002estimating}
Wagner, A. 2002.
\newblock Estimating coarse gene network structure from large-scale gene
  perturbation data.
\newblock {\em Genome Research\/}~{\em 12\/}(2), 309--315.

\bibitem[\protect\citeauthoryear{Wilson, Wang, Mucha, Bhamidi, Nobel,
  et~al.}{Wilson et~al.}{2014}]{wilson2013testing}
Wilson, J.~D., S.~Wang, P.~J. Mucha, S.~Bhamidi, A.~B. Nobel, et~al. 2014.
\newblock A testing based extraction algorithm for identifying significant
  communities in networks.
\newblock {\em Annals of Applied Statistics\/}~{\em 8\/}(3), 1853--1891.

\bibitem[\protect\citeauthoryear{Wolfe and Olhede}{Wolfe and
  Olhede}{2013}]{wolfe2013nonparametric}
Wolfe, P.~J. and S.~C. Olhede 2013.
\newblock Nonparametric graphon estimation.
\newblock {\em arXiv preprint arXiv:1309.5936\/}.

\bibitem[\protect\citeauthoryear{Yen and Larremore}{Yen and
  Larremore}{2020}]{yen2020community}
Yen, T.-C. and D.~B. Larremore 2020.
\newblock Community detection in bipartite networks with stochastic block
  models.
\newblock {\em Physical Review E\/}~{\em 102\/}(3), 032309.

\bibitem[\protect\citeauthoryear{Zanghi, Picard, Miele, Ambroise,
  et~al.}{Zanghi et~al.}{2010}]{zanghi2010strategies}
Zanghi, H., F.~Picard, V.~Miele, C.~Ambroise, et~al. 2010.
\newblock Strategies for online inference of model-based clustering in large
  and growing networks.
\newblock {\em The Annals of Applied Statistics\/}~{\em 4\/}(2), 687--714.

\bibitem[\protect\citeauthoryear{Zhang and Peixoto}{Zhang and
  Peixoto}{2020}]{zhang2020statistical}
Zhang, L. and T.~P. Peixoto 2020.
\newblock Statistical inference of assortative community structures.
\newblock {\em Physical Review Research\/}~{\em 2\/}(4), 043271.

\bibitem[\protect\citeauthoryear{Zhang, Levina, and Zhu}{Zhang
  et~al.}{2015}]{zhang2015estimating}
Zhang, Y., E.~Levina, and J.~Zhu 2015.
\newblock Estimating network edge probabilities by neighborhood smoothing.
\newblock {\em arXiv preprint arXiv:1509.08588\/}.

\bibitem[\protect\citeauthoryear{Zhao, Levina, Zhu, et~al.}{Zhao
  et~al.}{2012}]{zhao2012consistency}
Zhao, Y., E.~Levina, J.~Zhu, et~al. 2012.
\newblock Consistency of community detection in networks under degree-corrected
  stochastic block models.
\newblock {\em The Annals of Statistics\/}~{\em 40\/}(4), 2266--2292.

\end{thebibliography}

\clearpage

\section*{Appendix}

\subsection*{Appendix A: Derivation relating to Algorithm \ref{specClustAlg}, for the case of two co-communities}\label{specClustAlgDerivSupp}
Define $m$, $l$, $\mathbf{A}$, $\mathbf{B}$, $\mathbf{d}^{(X)}$, $\mathbf{d}^{(Y)}$, $d^{++}$, $g^{(X)}$, $g^{(Y)}$, $k^{(X)}$, $k^{(Y)}$, $\Psi$ and $Q_{XY}$ according to Definitions \ref{mainModelDef} - \ref{localCoModDef}. Specify that $k^{(X)}=k^{(Y)}=2$, that $T=2$, that $c_1=\left\{1,1\right\}$, and that $c_2=\left\{2,2\right\}$; i.e., that there are two co-communities, the first of which consists of $g^{(X)}_1$ paired with $g^{(Y)}_1$, and the second of which consists of $g^{(X)}_2$ paired with $g^{(Y)}_2$. Define co-community label vectors $\mathbf{s}$ and $\mathbf{r}$ for the $X$ and $Y$-nodes respectively, such that:
\begin{align}
    s_i= 
\begin{cases}\label{sVecDef}
    1, & \text{if } X\text{-node }i\text{ is in co-community 1},\\
    -1, & \text{if } X\text{-node }i\text{ is in co-community 2},
\end{cases}
\end{align}
and
\begin{align}
    r_j= 
\begin{cases}\label{rVecDef}
    1, & \text{if } Y\text{-node }j\text{ is in co-community 1},\\
    -1, & \text{if } Y\text{-node }j\text{ is in co-community 2}.
\end{cases}
\end{align}
Hence (referring to definition \ref{coModDef}):
\begin{equation*}
\Psi\left(C;G^{(X)},G^{(Y)};i,j\right)=\frac{1}{2}\left(s_ir_j+1\right),
\end{equation*}
and
\begin{equation*}
Q_{XY}=\frac{1}{2d^{++}}\sum_{i=1}^m\sum_{j=1}^lB_{ij}\left(s_ir_j+1\right).
\end{equation*}
Note that the rows of $\mathbf{B}$ sum to zero:
\begin{equation*}
\sum_{j=1}^lB_{ij}=\sum_{j=1}^lA_{ij}-\frac{d^{(X)}_i}{d^{++}}\sum_{j=1}^ld^{(Y)}_j=d^{(X)}_i-\frac{d^{(X)}_i}{d^{++}}\cdot d^{++}=0.
\end{equation*}
Also, the columns of $\mathbf{B}$ also sum to zero, by a similar argument. Hence:
\begin{equation}
Q_{XY}=\frac{1}{2d^{++}}\sum_{i=1}^m\sum_{j=1}^lB_{ij}s_ir_j.\label{coModDefNewm}
\end{equation}
When \citep{newman2013spectral} derives the properties of unipartite network community detection he relaxes the constraint that the co-community labels take the values of $\pm1$, to be able to arrive at an algebraic solution. Nodes are then assigned to one community or the other, according to their sign (in the two-community scenario). A similar relaxation is made here, allowing $s_i\in\mathbb{R}$ and $r_j\in\mathbb{R}$, subject also to the following elliptical constraints, which allow for degree heterogeneity as in the degree corrected stochastic blockmodel:
\begin{align}
\sum_{i=1}^md^{(X)}_is_i^2&=d^{++}\label{constraintX},\\
\sum_{j=1}^ld^{(Y)}_jr_j^2&=d^{++}\label{constraintY}.
\end{align}
In the extreme scenario, in which $s_i\in\left\{-1,1\right\}$ and $r_j\in\left\{-1,1\right\}$, these constraints are equivalent to $d^{++}=\sum_{i=1}^md^{(X)}_i=\sum_{j=1}^ld^{(Y)}_j$ (i.e., as per definition \ref{coModDef}). This relaxation is equivalent to saying that nodes may be partly in one group, and partly in another group (which also relates to the mixed-membership blockmodel \citep{airoldi2009mixed}). N.B., ultimately each node will be assigned entirely to only the group it is most strongly associated with (according to $s_i$ or $r_j$), and hence mixed membership does not occur in the final assignment of nodes to groups. For homogenous degree distributions, the constraints of equation \ref{constraintX} and \ref{constraintY} prevent the co-modularity from becoming arbitrarily large, as nodes are assigned many times over to many groups. For heterogenous degree distributions, the effect of the constraint is equivalent, except that the constraint is weighted to give importance to high-degree nodes. This is achieved by the constraints  of equation \ref{constraintX} and \ref{constraintY} restricting the weighted sum of the degrees (weighted by the assignment of nodes to groups) to be the equal to the total number of edges. 

We wish to find the community assignment vectors $\textbf{r}$ and $\textbf{s}$ which maximise the co-modularity, i.e., we want to maximise $Q_{XY}$ with respect to both $\textbf{r}$ and $\textbf{s}$. To do this, we employ the Lagrange multipliers $\lambda$ and $\mu$, and equate the derivatives to zero, N.B., the partial derivatives with respect to $s_{i'}$ and $r_{j'}$ are used as the derivatives are taken with respect to these individual $i'\in\left\{1,...,l\right\}$, and  $j'\in\left\{1,...,m\right\}$.
\begin{align*}
\frac{\partial}{\partial s_{i'}}\left[\sum_{i=1}^m\sum_{j=1}^lB_{ij}s_ir_j-\lambda\sum_{i=1}^md^{(X)}_is_i^2-\mu\sum_{j=1}^ld^{(Y)}_jr_j^2\right]&=0,\\
\text{and}\quad\frac{\partial}{\partial r_{j'}}\left[\sum_{i=1}^m\sum_{j=1}^lB_{ij}s_ir_j-\lambda\sum_{i=1}^md^{(X)}_is_i^2-\mu\sum_{j=1}^ld^{(Y)}_jr_j^2\right]&=0,
\end{align*}
\begin{equation}
\implies\quad\sum_{j=1}^lB_{ij}r_j-2\lambda d^{(X)}_is_i=0,\label{maxEQlineX}
\end{equation}
\begin{equation}
\text{and}\quad\sum_{i=1}^mB_{ij}s_i-2\mu d^{(Y)}_jr_j=0.\label{maxEQlineY}
\end{equation}
Hence, taking $\mathbf{D}^{(X)}$ and $\mathbf{D}^{(Y)}$ as the diagonal matrices with the degree vectors $\mathbf{d}^{(X)}$ and $\mathbf{d}^{(Y)}$ respectively on their leading diagonals,
\begin{equation}
\quad\mathbf{B}\mathbf{r}=2\lambda\mathbf{D}^{(x)}\mathbf{s}\label{SmatrixMinimised}
\end{equation}
\begin{equation}
\text{and}\quad\mathbf{B}^{\top}\mathbf{s}=2\mu\mathbf{D}^{(y)}\mathbf{r}.\label{RmatrixMinimised}
\end{equation}
Substituting for $\mathbf{s}$, equation \ref{RmatrixMinimised} in \ref{SmatrixMinimised}, gives:
\begin{equation}
\left(\mathbf{D}^{(y)}\right)^{-1}\mathbf{B}^{\top}\left(\mathbf{D}^{(x)}\right)^{-1}\mathbf{B}\mathbf{r}=4\lambda\mu\mathbf{r},\label{simpMatrixMin}
\end{equation}
\begin{equation*}
\implies\left(\mathbf{D}^{(y)}\right)^{-1/2}\mathbf{B}^{\top}\left(\mathbf{D}^{(x)}\right)^{-1/2}\left(\mathbf{D}^{(x)}\right)^{-1/2}\mathbf{B}\left(\mathbf{D}^{(y)}\right)^{-1/2}\mathbf{r}=4\lambda\mu\mathbf{r},
\end{equation*}
\begin{equation}
\implies\left(\left(\mathbf{D}^{(x)}\right)^{-1/2}\mathbf{B}\left(\mathbf{D}^{(y)}\right)^{-1/2}\right)^{\top}\left(\left(\mathbf{D}^{(x)}\right)^{-1/2}\mathbf{B}\left(\mathbf{D}^{(y)}\right)^{-1/2}\right)\mathbf{r}=4\lambda\mu\mathbf{r},\label{eigenEQ1}
\end{equation}
\begin{equation}
\implies\mathbf{M}^{\top}\mathbf{M}\mathbf{r}=4\lambda\mu\mathbf{r},\label{lambdaEigenEq}
\end{equation}
where
\begin{equation*}
\mathbf{M}=\left(\mathbf{D}^{(x)}\right)^{-1/2}\mathbf{B}\left(\mathbf{D}^{(y)}\right)^{-1/2}.
\end{equation*}
By an identical argument, substituting \ref{SmatrixMinimised} in \ref{RmatrixMinimised} and re-arranging equivalently,
\begin{equation}
\mathbf{M}\mathbf{M}^{\top}\mathbf{s}=4\lambda\mu\mathbf{s}.\label{muEigenEq}
\end{equation}
Hence, $\mathbf{s}$ and $\mathbf{r}$ are eigenvectors of $\mathbf{M}\mathbf{M}^{\top}$ and $\mathbf{M}^{\top}\mathbf{M}$ respectively, with $4\lambda\mu$ the corresponding eigenvalue in both cases. Therefore, $\mathbf{s}$ and $\mathbf{r}$ are left and right singular vectors respectively of:
\begin{equation*}
\mathbf{M}=\left(\mathbf{D}^{(x)}\right)^{-1/2}\mathbf{B}\left(\mathbf{D}^{(y)}\right)^{-1/2},
\end{equation*}
with corresponding singular value $2\sqrt{\lambda\mu}$.

Multiplying equation \ref{maxEQlineX} by $s_i/2d^{++}$, summing over $i$ and referring to equation \ref{constraintX} gives:
\begin{align*}
\frac{1}{2d^{++}}\sum_{i=1}^m\sum_{j=1}^lB_{ij}s_ir_j=\frac{2\lambda}{2d^{++}}\sum_{i=1}^md^{(X)}_is_i^2=\frac{2\lambda\cdot d^{++}}{2d^{++}}=\lambda,
\end{align*}
hence referring to equation \ref{coModDefNewm}, we get:
\begin{equation}
Q_{XY}=\lambda.\label{optCoModLambda}
\end{equation}
Then equivalently multiplying equation \ref{maxEQlineY} by $r_j/2d^{++}$, summing over $j$ and referring to equation \ref{constraintY}, and then referring to equation \ref{coModDefNewm} gives: 
\begin{equation}
Q_{XY}=\mu.\label{optCoModMu}
\end{equation}
Therefore, referring again to equations \ref{lambdaEigenEq} and \ref{muEigenEq}, the maximum modularity solution is for the left and right singular vectors of $M$ which correspond to the greatest singular value $2\lambda$.

Now substituting equation \ref{Bdef} in equation \ref{RmatrixMinimised}, we get:
\begin{equation*}
\mathbf{s}^{\top}\left(\mathbf{A}-\frac{1}{d^{++}}\mathbf{d}^{(X)}\left(\mathbf{d}^{(Y)}\right)^{\top}\right)=2\mu\mathbf{r}^{\top}\mathbf{D}^{(y)},
\end{equation*}
\begin{equation}
\implies\mathbf{s}^{\top}\mathbf{A}=\frac{1}{d^{++}}\mathbf{s}^{\top}\mathbf{d}^{(X)}\left(\mathbf{d}^{(Y)}\right)^{\top}+2\mu\mathbf{r}^{\top}\mathbf{D}^{(y)}.\label{optCoModIntEq}
\end{equation}
Post-multiplying equation \ref{optCoModIntEq} by $\mathbf{1}=(1,1,1...)$ leads to:
\begin{equation*}
\mathbf{s}^{\top}\mathbf{d}^{(X)}=\frac{1}{d^{++}}\mathbf{s}^{\top}\mathbf{d}^{(X)}\cdot d^{++}+2\mu\mathbf{r}^{\top}\mathbf{d}^{(Y)}
\end{equation*}
\begin{equation*}
\therefore\quad\mu\mathbf{r}^{\top}\mathbf{d}^{(Y)}=0.
\end{equation*}
Assuming that there is co-community structure present in $\mathbf{A}$, there must be positive co-modularity, i.e., $Q_{XY}>0\enskip\implies\enskip\mu>0$ (referring back to equation \ref{optCoModMu}), and therefore $\mathbf{r}^{\top}\mathbf{d}^{(Y)}=0$. By an identical argument, also $\mathbf{s}^{\top}\mathbf{d}^{(X)}=0$. Therefore, for eigenvectors $\mathbf{r}$ corresponding to $Q_{XY}>0$,
\begin{equation*}
\mathbf{B}\mathbf{r}=\left(\mathbf{A}-\frac{1}{d^{++}}\mathbf{d}^{(X)}\left(\mathbf{d}^{(Y)}\right)^{\top}\right)\mathbf{r}=\mathbf{A}\mathbf{r}
\end{equation*}
and so to find these eigenvectors with $Q_{XY}$ maximised, instead of equation \ref{eigenEQ1} we can consider
\begin{equation}
\left(\left(\mathbf{D}^{(x)}\right)^{-1/2}\mathbf{A}\left(\mathbf{D}^{(y)}\right)^{-1/2}\right)^{\top}\left(\left(\mathbf{D}^{(x)}\right)^{-1/2}\mathbf{A}\left(\mathbf{D}^{(y)}\right)^{-1/2}\right)\mathbf{r}=\left(2\lambda\right)^2\mathbf{r}\label{AeigenEQ}
\end{equation}
which, referring back to equation \ref{coLaplacian}, can be written in terms of the co-Laplacian $\mathbf{L}_{XY}$ as:
\begin{align*}
\mathbf{L}_{XY}^{\top}\mathbf{L}_{XY}\mathbf{r}&=\left(2\lambda\right)^2\mathbf{r}.
\end{align*}
By identical argument, we can also write:
\begin{equation}
\left(\left(\mathbf{D}^{(x)}\right)^{-1/2}\mathbf{A}\left(\mathbf{D}^{(y)}\right)^{-1/2}\right)\left(\left(\mathbf{D}^{(x)}\right)^{-1/2}\mathbf{A}\left(\mathbf{D}^{(y)}\right)^{-1/2}\right)^{\top}\mathbf{s}=\left(2\lambda\right)^2\mathbf{s}\label{AeigenEQ2}
\end{equation}
and
\begin{align*}
\mathbf{L}_{XY}\mathbf{L}_{XY}^{\top}\mathbf{s}&=\left(2\lambda\right)^2\mathbf{s}.
\end{align*}
Hence, the co-Laplacian $\mathbf{L}_{XY}$ has left and right singular vectors $\mathbf{s}$ and $\mathbf{r}$ respectively, with corresponding singular values $2\lambda$. It can be seen that equation \ref{AeigenEQ} has the eigenvector $\mathbf{1}=(1,1,1,...)$, as follows:
\begin{align*}
\left(\left(\mathbf{D}^{(x)}\right)^{-1/2}\mathbf{A}\left(\mathbf{D}^{(y)}\right)^{-1/2}\right)^{\top}\left(\left(\mathbf{D}^{(x)}\right)^{-1/2}\mathbf{A}\left(\mathbf{D}^{(y)}\right)^{-1/2}\right)\mathbf{1}&=\left(2\lambda\right)^2\mathbf{1}\\
\implies\left(\mathbf{D}^{(y)}\right)^{-1}\mathbf{A}^{\top}\left(\mathbf{D}^{(x)}\right)^{-1}\mathbf{A}\mathbf{1}&=\left(2\lambda\right)^2\mathbf{1}\\
\implies\left(\mathbf{D}^{(y)}\right)^{-1}\mathbf{A}^{\top}\left(\mathbf{D}^{(x)}\right)^{-1}\mathbf{d}^{(X)}&=\left(2\lambda\right)^2\mathbf{1}\\
\implies\left(\mathbf{D}^{(y)}\right)^{-1}\mathbf{A}^{\top}\mathbf{1}&=\left(2\lambda\right)^2\mathbf{1}\\
\implies\left(\mathbf{D}^{(y)}\right)^{-1}\mathbf{d}^{(Y)}&=\left(2\lambda\right)^2\mathbf{1}\\
\mathbf{1}&=\left(2\lambda\right)^2\mathbf{1}
\end{align*}
and hence the corresponding eigenvalue is $\left(2\lambda\right)^2=1$, which by the Perron-Frobenius theorem, must be the greatest eigenvalue \citep{hom1991topics,newman2013spectral}. An identical argument can also be applied to $\textbf{s}$ in equation \ref{AeigenEQ2}. This means that the greatest singular value $2\lambda=1$ corresponds to these left and right singular vectors which are both $\mathbf{1}$ (of lengths $m$ and $l$ respectively), however such singular vectors do not satisfy $\mathbf{r}^{\top}\mathbf{d}^{(Y)}=0$ and $\mathbf{s}^{\top}\mathbf{d}^{(X)}=0$. Therefore, to maximise the co-modularity in the case of two co-communities, we should divide the $X$ and $Y$-nodes according to the left and right singular vectors respectively which correspond to the second greatest singular value.

The above explains how Algorithm \ref{specClustAlg} works for the case of two co-communities. An equivalent extension to $k$ communities has been made in the unipartite community detection setting \citep{riolo2014first}. To do so, the community labels are identified with the vertices of $k-1$ simplicies, i.e., for detection of 3 communities, the co-community labels would be the vertices of a triangle. Relaxing constraints equivalent to equations \ref{constraintX} and \ref{constraintY} means allowing the nodes to move away from the vertices of the simplex. This amounts to clustering the nodes in the space of the eigenvectors corresponding to the 2\textsuperscript{nd} to $k$\textsuperscript{th} greatest eigenvalues of the Laplacian $\mathbf{L}$. This clustering is conventionally done using $k$-means. The reader is referred to \citep{riolo2014first} for the detailed technical derivations relating to this. A similar extension can naturally be made in this co-community detection setting. To detect $k^{(X)}$ $X$-node groupings, and $k^{(Y)}$ $Y$-node groupings, the $X$ and $Y$-nodes can be separately clustered (using $k$-means independently for the $X$ and $Y$-nodes) in the spaces of the left and right singular vectors (respectively) corresponding to the 2\textsuperscript{nd} to $k^{(X)}$\textsuperscript{th} and 2\textsuperscript{nd} to $k^{(Y)}$\textsuperscript{th} greatest singular values, respectively, of the singular value decomposition of the co-Laplacian $\mathbf{L}_{XY}$.

\subsection*{Appendix B: Proof of Proposition \ref{poisLprop}, for the case of two co-communities}
For the case of two co-communities, with $\theta_{\mathrm{in}}$ and $\theta_{\mathrm{out}}$ defined according to equation \ref{poissonCons}, with the co-community labels $r_i$ and $s_j$ defined as in Appendix A / section \ref{specClustAlgDerivSupp} (equations \ref{sVecDef} and \ref{rVecDef}), and with $G^{(X)}$ and $G^{(Y)}$ defined according to Definition \ref{mainModelDef}, we note (equivalently to \citep{newman2013spectral}) that:
\begin{align}
\theta_{z^{(X)}(i),z^{(Y)}(j)}&=\frac{1}{2}\left(\theta_{\mathrm{in}}+\theta_{\mathrm{out}}+r_is_j\left(\theta_{\mathrm{in}}-\theta_{\mathrm{out}}\right)\right),\label{noteEQ} \\
\text{and\quad} \ln{\left(\theta_{z^{(X)}(i),z^{(Y)}(j)}\right)}&=\frac{1}{2}\left(\ln{\left(\theta_{\mathrm{in}}\theta_{\mathrm{out}}\right)}+r_is_j\ln{\left(\frac{\theta_{\mathrm{in}}}{\theta_{\mathrm{out}}}\right)}\right),\label{noteEQlog}
\end{align}
We note that equations \ref{noteEQ} and \ref{noteEQlog} only hold because $s_i\in\left\{-1,1\right\}$ and $r_j\in\left\{-1,1\right\}$. Substituting equations \ref{noteEQ} and \ref{noteEQlog} into equation \ref{modelEll}, and estimating the node-specific connectivity parameters $\boldsymbol{\pi}^{(X)}$ and $\boldsymbol{\pi}^{(Y)}$ by the degree distributions $\mathbf{d}^{(X)}$ and $\mathbf{d}^{(Y)}$ leads to the profile likelihood (using the Poisson approximation of \cite{perry2012null}):
\begin{multline*}
\ell\left(\boldsymbol{\theta};\mathbf{d}^{(X)},\mathbf{d}^{(Y)};G^{(X)},G^{(Y)}\right)=\sum_{i=1}^m\sum_{j=1}^l\left[\frac{A_{ij}}{2}\left(\ln{\left(\theta_{\mathrm{in}}\theta_{\mathrm{out}}\right)}+r_is_j\ln{\left(\frac{\theta_{\mathrm{in}}}{\theta_{\mathrm{out}}}\right)}\right)\right.\\
\left.-\frac{d^{(X)}_id^{(Y)}_j}{2}\left(\theta_{\mathrm{in}}+\theta_{\mathrm{out}}+r_is_j\left(\theta_{\mathrm{in}}-\theta_{\mathrm{out}}\right)\right)\right]
\end{multline*}
\begin{multline*}
\implies\quad\ell\left(\boldsymbol{\theta};\mathbf{d}^{(X)},\mathbf{d}^{(Y)};G^{(X)},G^{(Y)}\right)=\frac{1}{2}\sum_{i=1}^m\sum_{j=1}^l\left[A_{ij}\ln{\left(\theta_{\mathrm{in}}\theta_{\mathrm{out}}\right)}-d^{(X)}_id^{(Y)}_j\left(\theta_{\mathrm{in}}+\theta_{\mathrm{out}}\right)\right.\\ 
\left.+\ln{\left(\frac{\theta_{\mathrm{in}}}{\theta_{\mathrm{out}}}\right)}\left(A_{ij}-d^{(X)}_id^{(Y)}_j\cdot\frac{\theta_{\mathrm{in}}-\theta_{\mathrm{out}}}{\ln\theta_{\mathrm{in}}-\ln\theta_{\mathrm{out}}}\right)s_ir_j\right].
\end{multline*}
We seek to maximise $\ell\left(\boldsymbol{\theta};\mathbf{d}^{(X)},\mathbf{d}^{(Y)};G^{(X)},G^{(Y)}\right)$ with respect to $G^{(X)}$ and $G^{(Y)}$ by choosing the co-community labels $s_i$ and $r_j$. Therefore, we can drop the terms constant in $s_i$ and $r_j$ to give:
\begin{equation*}
\widetilde{\ell}\left(\boldsymbol{\theta};\mathbf{d}^{(X)},\mathbf{d}^{(Y)};G^{(X)},G^{(Y)}\right)=\sum_{i=1}^m\sum_{j=1}^l\left(A_{ij}-d^{(X)}_id^{(Y)}_j\cdot\frac{\theta_{\mathrm{in}}-\theta_{\mathrm{out}}}{\ln\theta_{\mathrm{in}}-\ln\theta_{\mathrm{out}}}\right)s_ir_j,\\
\end{equation*}
and defining:
\begin{equation*}
\eta=\frac{\theta_{\mathrm{in}}-\theta_{\mathrm{out}}}{\ln\theta_{\mathrm{in}}-\ln\theta_{\mathrm{out}}},
\end{equation*}
we therefore have:
\begin{equation}
\widetilde{\ell}\left(\boldsymbol{\theta};\mathbf{d}^{(X)},\mathbf{d}^{(Y)};G^{(X)},G^{(Y)}\right)=\sum_{i=1}^m\sum_{j=1}^l\left(A_{ij}-\eta d^{(X)}_id^{(Y)}_j\right)s_ir_j,\label{modelEllTildeIntermediate}
\end{equation}
which we note as equivalent to equation 22 in \citep{newman2013spectral}. Proceeding similarly to that work, by applying to equation \ref{modelEllTildeIntermediate} the constraints of equations \ref{constraintX} and \ref{constraintY} with Lagrange multipliers $\lambda$ and $\mu$ and differentiating and equating to zero, we get:
\begin{align*}
\frac{\partial}{\partial s_{i'}}\left[\sum_{i=1}^m\sum_{j=1}^l\left(A_{ij}-\eta d^{(X)}_id^{(Y)}_j\right)s_ir_j-\lambda\sum_{i=1}^md^{(X)}_is_i^2-\mu\sum_{j=1}^ld^{(Y)}_jr_j^2\right]&=0,\\
\frac{\partial}{\partial r_{j'}}\left[\sum_{i=1}^m\sum_{j=1}^l\left(A_{ij}-\eta d^{(X)}_id^{(Y)}_j\right)s_ir_j-\lambda\sum_{i=1}^md^{(X)}_is_i^2-\mu\sum_{j=1}^ld^{(Y)}_jr_j^2\right]&=0,
\end{align*}
\begin{align*}
\implies\quad\sum_{j=1}^l\left(A_{ij}-\eta d^{(X)}_id^{(Y)}_j\right)r_j-2\lambda d^{(X)}_is_i&=0,\\
\text{and}\quad\sum_{i=1}^m\left(A_{ij}-\eta d^{(X)}_id^{(Y)}_j\right)s_i-2\mu d^{(Y)}_jr_j&=0,
\end{align*}
and therefore also recalling the definitions of $\mathbf{D}^{(X)}$ and $\mathbf{D}^{(Y)}$ as the diagonal matrices with the degree vectors $\mathbf{d}^{(X)}$ and $\mathbf{d}^{(Y)}$ respectively on their leading diagonals,
\begin{align}
\quad\left(\mathbf{A}-\eta\boldsymbol{\mathbf{d}}^{(X)}\left(\boldsymbol{\mathbf{d}}^{(Y)}\right)^{\top}\right)\mathbf{r}&=2\lambda\mathbf{D}^{(X)}\mathbf{s},\label{maximSell}\\
\text{and}\quad\left(\mathbf{A}^{\top}-\eta\boldsymbol{\mathbf{d}}^{(Y)}\left(\boldsymbol{\mathbf{d}}^{(X)}\right)^{\top}\right)\mathbf{s}&=2\mu\mathbf{D}^{(Y)}\mathbf{r}.\label{maximRell}
\end{align}
Combining equations \ref{maximSell} and \ref{maximRell} by substituting for $\mathbf{s}$ and $\mathbf{r}$, and following simplification identical to equations \ref{simpMatrixMin} to \ref{eigenEQ1}, gives:
\begin{align*}
\mathbf{W}^{\top}\mathbf{W}\mathbf{r}&=4\lambda\mu\mathbf{r},\\
\text{and}\quad\mathbf{W}\mathbf{W}^{\top}\mathbf{s}&=4\lambda\mu\mathbf{s},
\end{align*}
where 
\begin{equation*}
\mathbf{W}=\left(\mathbf{D}^{(X)}\right)^{-1/2}\left(\mathbf{A}-\eta\boldsymbol{\mathbf{d}}^{(X)}\left(\boldsymbol{\mathbf{d}}^{(Y)}\right)^{\top}\right)\left(\mathbf{D}^{(Y)}\right)^{-1/2}.
\end{equation*}
Hence $\mathbf{s}$ and $\mathbf{r}$ are left and right singular vectors of the singular value decomposition of $\mathbf{W}$, again with corresponding singular values $4\lambda\mu$. Defining $\mathbf{1}=(1,1,1...)$, and noting that $\mathbf{1}\boldsymbol{\mathbf{d}}^{(X)}\left(\boldsymbol{\mathbf{d}}^{(Y)}\right)^{\top}=d^{++}\left(\boldsymbol{\mathbf{d}}^{(Y)}\right)^{\top}$ etc, we can see that pre-multiplying \ref{maximSell} and \ref{maximRell} by $\mathbf{1}$ leads to:
\begin{align}
\mathbf{r}^{\top}\boldsymbol{\mathbf{d}}^{(Y)}\left(1-d^{++}\eta\right)&=2\lambda\mathbf{s}^{\top}\boldsymbol{\mathbf{d}}^{(X)},\label{eigenEllS}\\
\text{and}\quad\mathbf{s}^{\top}\boldsymbol{\mathbf{d}}^{(X)}\left(1-d^{++}\eta\right)&=2\mu\mathbf{r}^{\top}\boldsymbol{\mathbf{d}}^{(Y)}.\label{eigenEllR}
\end{align}
Substituting for $\mathbf{s}^{\top}\boldsymbol{\mathbf{d}}^{(X)}$ and $\mathbf{r}^{\top}\boldsymbol{\mathbf{d}}^{(Y)}$, equation \ref{eigenEllR} in equation \ref{eigenEllS} and {\it vice-versa}, gives:
\begin{align*}
\mathbf{s}^{\top}\boldsymbol{\mathbf{d}}^{(X)}\left[\left(1-d^{++}\eta\right)^2-4\mu\lambda\right]&=0,\\
\text{and}\quad\mathbf{r}^{\top}\boldsymbol{\mathbf{d}}^{(Y)}\left[\left(1-d^{++}\eta\right)^2-4\mu\lambda\right]&=0,
\end{align*}
and therefore because $\left(1-d^{++}\eta\right)^2-4\mu\lambda$ is not guaranteed to be zero, 
\begin{align*}
\mathbf{s}^{\top}\boldsymbol{\mathbf{d}}^{(X)}&=0,\\
\text{and}\quad\mathbf{r}^{\top}\boldsymbol{\mathbf{d}}^{(Y)}&=0.
\end{align*}
Therefore, equations \ref{maximSell} and \ref{maximRell} reduce to:
\begin{align*}
\mathbf{A}\mathbf{r}&=2\lambda\mathbf{D}^{(x)}\mathbf{s}\\
\text{and}\quad\mathbf{A}^{\top}\mathbf{s}&=2\mu\mathbf{D}^{(y)}\mathbf{r},
\end{align*}
and again combining these equations by substituting for $\mathbf{s}$ and $\mathbf{r}$ and following equivalent simplification to equations \ref{simpMatrixMin} to \ref{eigenEQ1}, we hence find that $\mathbf{s}$ and $\mathbf{r}$ are left and right singular vectors of the co-Laplacian (equation \ref{coLaplacian}).
Therefore, the choice of the co-community labels $\mathbf{s}$ and $\mathbf{r}$ which maximises the model likelihood specified in equation \ref{modelEll}, subject also to the constraint of equation \ref{poissonCons}, is equivalent to the maximum co-modularity assignment obtained via Algorithm \ref{specClustAlg}.

\subsection*{Appendix C: Selecting the Number of Co-communities}\label{selNumCom}
In order to use Algorithm \ref{specClustAlg} to carry out co-community detection, we must specify the number of $X$-node groupings $k^{(X)}$, and the number of $Y$-node groupings $k^{(Y)}$. The following procedure for estimating optimal numbers of these groupings was developed as part of the author's PhD under the guidance of Sofia Olhede.

If we want to let the network grow, it would be impractical to fully specify more complicated versions of the parametric model of Definition \ref{mainModelDef}, which completely account for all effects. Instead, we can make a non-parametric generalisation of this model incorporating more smoothing, based on the notion of the graphon. The graphon is a latent, smooth function which sets the probability between each pair of nodes, of a connection forming between that pair of nodes \citep{wolfe2013nonparametric}. In this setting, the graphon is not symmetric, due to the two different types of nodes modelled.
\begin{Definition}\label{anisGraphModelDef}
For the Lipschitz-continuous graphon function $f\in L\left((0,1)^2\right)$, with $\mathbf{A}$ defined according to Definition \ref{mainModelDef}, define connectivity functions $\phi^{(X)}\in L\left(0,1\right)$ and $\phi^{(Y)}\in L\left(0,1\right)$, and define latent orderings $\xi^{(X)}_i\overset{i.i.d.}{\sim}\mathcal{U}\left(0,1\right)$ and (independently) $\xi^{(Y)}_j\overset{i.i.d.}{\sim}\mathcal{U}\left(0,1\right)$ on the graphon margins of $X$ and $Y$-nodes $i\in\{1,...,m\}$ and $j\in\{1,...,l\}$ respectively. Then,
\begin{equation}
\mathbb{E}\left(A_{ij}\right)=f\left(\xi^{(X)}_i,\xi^{(Y)}_j\right)\cdot\phi^{(X)}\left(\xi^{(X)}_i\right)\cdot\phi^{(Y)}\left(\xi^{(Y)}_j\right). \label{graphonExpAdj}
\end{equation}
\end{Definition}
\noindent
The graphon $f$ (Definition \ref{anisGraphModelDef}) can be considered an infinite-dimensional equivalent to $\theta_{p,q}$ (Definition \ref{mainModelDef}), up to a re-ordering of the nodes (after the Aldous-Hoover theorem \citep{aldous1985exchangeability}). The connectivity functions $\phi^{(X)}$ and $\phi^{(Y)}$ (Definition \ref{anisGraphModelDef}) are then similarly equivalent to the node-specific connectivity parameters $\boldsymbol{\pi}^{(X)}$ and $\boldsymbol{\pi}^{(Y)}$ (Definition \ref{mainModelDef} for the degree corrected stochastic co-blockmodel). These functions $\phi^{(X)}$ and $\phi^{(Y)}$ model the general variability of connectivity strength throughout the network, whereas the graphon $f$ models the tendency for regions of the network to aggregate into specific co-communities. The model of Definition \ref{anisGraphModelDef} is a more general model which is specified similarly for any network size. Thus, equation \ref{graphonExpAdj} contains redundancy, and hence as the networks we consider here are of fixed size, the degree corrected stochastic co-blockmodel (Definition \ref{mainModelDef}) may be a more parsimonious choice. To estimate the generating mechanism of a bipartite network stably, Definition \ref{anisGraphModelDef} must be replaced by a model with a limited number of parameters, i.e., Definition \ref{mainModelDef}.

The network histogram method of fitting the stochastic blockmodel \citep{olhede2014net} in the unipartite/symmetric community detection setting provides a rule-of-thumb method for selecting the optimal number of communities, or blocks, in the model. However we note that this rule-of-thumb method may not result in a perfect match between the number of communities and the size of blocks, as spectral clustering may not result in equal size communities. Fitted in this way, the blockmodel is a valid representation of a network, whatever the generating mechanism of that network, as long as this generating mechanism results in an exchangeable network. The network histogram approximates the graphon, which is a continuous function: the nodes correspond to discrete locations along the graphon margins, ordered in an optimal way to satisfy the smoothness requirement of the graphon. The graphon oracle \citep{wolfe2013nonparametric,olhede2014net} defines a good ordering of the nodes, according to graphon smoothness, and co-association patterns. This information is not available in practice, but it can be used to bound the mean integrated squared error of the network histogram approximation to the graphon. This ordering naturally corresponds to community assignments, and the number of communities, or blocks, is determined by the smoothness of the graphon. An intuition for this is by analogy with a wave: if there are many peaks over a fixed distance (i.e., short wavelength), the maximum gradient of the wave will be large, whereas if there are few peaks over the same fixed distance (i.e., long wavelength), the maximum gradient will be small. Similarly, the more communities, or peaks, that there are in the graphon, the greater the maximum gradient of the graphon will be, and, correspondingly, the less smooth it will be. 

\subsubsection*{Finding the optimal numbers of $X$ and $Y$-node groupings}
In this section we define the anisotropic graphon, which allows us to determine an optimal number of $X$ and $Y$-node groupings, $k^{(X)}$ and $k^{(Y)}$, from which co-communities can be identified. This relates closely to the network histogram method in the symmetric unipartite community detection setting \citep{olhede2014net}. In the unipartite community detection setting, the graphon is a symmetric limit object bounded on $(0,1)^2$. It is symmetric because the in that setting, the set of $X$-nodes is the same as the set of $Y$-nodes, and hence the smoothness is the same with respect to the corresponding orthogonal directions on the graphon. In contrast, in this co-community detection setting the graphon is asymmetric, having different smoothnesses with respect to the $X$ and $Y$-nodes. Hence, we refer to this as the `anisotropic graphon', which is similarly a limit object bounded on $(0,1)^2$. To aid our analyses, we can stretch the anisotropic graphon so that it has the same smoothness with respect to the $X$-nodes, and with respect to the $Y$-nodes. It is easy to see that such a transformation exists for all anisotropic graphons. We refer to the result of stretching the anisotropic graphon in this way, as the `equi-smooth graphon'. Without loss of generality, this transformation can be expressed as a stretch of scale-factor $\gamma$ with respect to the $X$-nodes, and a simultaneous stretch of scale-factor $1/\gamma$ with respect to the $Y$-nodes. We refer to $\gamma$ as the anisotropy factor. This is formalised as follows.
\begin{Definition}\label{anGraGammaDef}
For the Lipschitz-continuous anisotropic graphon $f\in L\left((0,1)^2\right)$ defined according to Definition \ref{anisGraphModelDef}, let the anisotropy factor $\gamma$ define the linear-stretch transformation which maps $f$ onto the Lipschitz-continuous equi-smooth graphon $\widetilde{f}\in L\left((0,\gamma)\times(0,1/\gamma)\right)$. Then,
\begin{equation}
f(x,y)=\widetilde{f}\left(\gamma x,y/\gamma\right).\label{graphonMainEq}
\end{equation}
\end{Definition}
\noindent
Lipschitz-continuity, in this context, means that the smoothness of the graphon (anisotropic or equi-smooth) is upper-bounded, and we use this bound to calculate the optimal number of $X$ and $Y$-node groupings. We note that as the graphon is asymmetric in the bipartite case, this assignation is still suitable.

To determine the optimal number of $X$ and $Y$-node groupings, $k^{(X)}$ and $k^{(Y)}$, assuming a fixed block-size to do the mapping, we set these $k^{(X)}$ and $k^{(Y)}$ so as to minimise the mean integrated squared error (MISE) of the blockmodel approximation of the graphon. Following a methodology which is closely related to the network histogram estimator in the symmetric (unipartite) community detection setting \citep{olhede2014net}, making use of the graphon oracle estimator, an upper bound can be calculated on this MISE, from a bias-variance decomposition, as follows:
\begin{Lemma}\label{MISElem}
With $\mathbf{A}$, $m$, $l$, $g^{(X)}\in G^{(X)}$, and $g^{(Y)}\in G^{(Y)}$ defined according to Definition \ref{mainModelDef}, let $\rho$ be a deterministic scaling constant which specifies the expected number of edges in the network, such that:
\begin{equation*}
\rho=\mathbb{E}\left(\frac{1}{ml}\sum_{j=1}^{l}\sum_{i=1}^{m}A_{ij}\right),
\end{equation*}
and define piecewise block-approximations to the adjacency matrix, for each pairing of a set of $X$-nodes $g^{(X)}$ with a set of $Y$-nodes $g^{(Y)}$, as:
\begin{equation*}
\bar{A}_{p,q}=\frac{\sum_{i\in g^{(X)}_p,j\in g^{(Y)}_q}A_{ij}}{\left|g^{(X)}_p\right|\left|g^{(Y)}_q\right|},
\end{equation*}
where $\left|\cdot\right|$ represents cardinality. With $z^{(X)}$ and $z^{(Y)}(j)$ defined according to Definition \ref{mainModelDef}, $\xi^{(X)}$ and $\xi^{(Y)}$ defined according to Definition \ref{anisGraphModelDef}, and $f$ defined according to Definition \ref{anGraGammaDef}, define alternative map functions $\widetilde{z}^{(X)}(i')$, $i'\in\left\{1,...,m\right\}$, and $\widetilde{z}^{(Y)}(j')$, $j'\in\left\{1,...,l\right\}$. These map functions take the ordered locations of the $X$ and $Y$-nodes respectively along the graphon margins, as specified by $\xi^{(X)}$ and $\xi^{(Y)}$, and return the corresponding $X$ and $Y$-node groupings, such that $\widetilde{z}^{(X)}\left(\ceil[\Big]{m\cdot\xi^{(X)}_i}\right)=z^{(X)}(i)$, and $\widetilde{z}^{(Y)}\left(\ceil[\Big]{l\cdot\xi^{(Y)}_j}\right)=z^{(Y)}(j)$. Define the graphon oracle estimator as:
\begin{equation}
\hat{f}(x,y)=\hat{\rho}^{-1}\bar{A}_{\widetilde{z}^{(X)}(\lceil lx\rceil),\widetilde{z}^{(Y)}(\lceil my\rceil)},
\end{equation}
and let:
\begin{equation}
\iint_{(0,1)^2}f(x,y)dx\ dy=1.\label{graphonIntConstr}
\end{equation}
With $\widetilde{f}$ and $\gamma$ defined as in Definition \ref{anGraGammaDef}, let $\widetilde{M}$ be the maximum gradient of $\widetilde{f}$, and let $h^{(X)}$ and $h^{(Y)}$ be `bandwidth' parameters with respect to the $X$ and $Y$ nodes respectively. Then, the graphon oracle upper bound on the MISE of the blockmodel estimate of the graphon function $\hat{f}$ is:
\begin{multline}
\mathrm{MISE\left(\hat{f}\right)}\leq \widetilde{M}^2\left\{\gamma^2\cdot\frac{\left(h^{(X)}\right)^2}{m^2}+\frac{1}{\gamma^2}\cdot\frac{\left(h^{(Y)}\right)^2}{l^2}\right\}\\+2\widetilde{M}^2\left\{\gamma^2\cdot\frac{1}{4m}+\frac{1}{\gamma^2}\cdot\frac{1}{4l}\right\}\left\{1+o\left(1\right)\right\}+\frac{1}{\rho\cdot h^{(X)}\cdot\ h^{(Y)}}\left\{1+o\left(1\right)\right\}.\label{MISEexpr}
\end{multline}
\end{Lemma}
\begin{proof}
See Appendix D.
\end{proof}
\noindent
As would be expected from the form of the anisotropic graphon, this expression directly captures the resolution obtained in each axis. We note that the ordering of the nodes along the adjacency matrix margins is not necessarily the same as in their mappings along the graphon margins. Thus, we need to specify how nodes map to the groupings $g^{(X)}$ and $g^{(Y)}$ in a different way for the graphon, as compared to the adjacency matrix. This difference is accounted for by using different mapping functions: $\widetilde{z}^{(X)}(i')$ and $\widetilde{z}^{(Y)}(j')$ for the graphon, and $z^{(X)}(i)$ and $z^{(Y)}(j)$ for the adjacency matrix. I.e., $\widetilde{z}^{(X)}(i')$ and $\widetilde{z}^{(Y)}(j')$ are required to specify the (contiguous) ranges and locations of the $X$ and $Y$-node groupings $g^{(X)}$ and $g^{(Y)}$ on the graphon margins, and equivalently $z^{(X)}(i)$ and $z^{(Y)}(j)$ for their (non-contiguous) locations on the adjacency matrix margins.

Using the MISE formulation of Lemma \ref{MISElem}, we can estimate the optimal numbers of $X$ and $Y$-node groupings, $k^{(X)}$ and $k^{(Y)}$.
\begin{Proposition}\label{optClustProp}
With $m$ and $l$ defined as in Definition \ref{mainModelDef}, and $\widetilde{M}$ and $\rho$ defined as in Lemma \ref{MISElem}, the optimal number of $X$ and $Y$-node groupings, $k^{(X)}$ and $k^{(Y)}$ respectively, are:
\begin{equation}
k^{(X)}=\gamma\cdot\left(ml\right)^{\frac{1}{4}}\cdot\left(2\rho \widetilde{M}^2\right)^{\frac{1}{4}}\label{estKx}
\end{equation}
and
\begin{equation}
k^{(Y)}=\frac{1}{\gamma}\cdot\left(ml\right)^{\frac{1}{4}}\cdot\left(2\rho \widetilde{M}^2\right)^{\frac{1}{4}}.\label{estKy}
\end{equation}
\end{Proposition}
\begin{proof}
The proof of this proposition is developed from the equivalent proof for the case of the isotropic graphon (corresponding to community detection in unipartite networks) \citep{olhede2014net}. The optimal bandwiths $h^{(X)*}$ and $h^{(Y)*}$ can be found by optimising the expression for the MISE of equation \ref{MISEexpr} with respect to $h^{(X)}$ and to $h^{(Y)}$ and setting to zero, and combining the resulting equations. To calculate $k^{(X)}$ and $k^{(Y)}$, substitute these optimal bandwiths $h^{(X)*}$ and $h^{(Y)*}$ into $k^{(X)}=m/h^{(X)*}$ and $k^{(Y)}=l/h^{(Y)*}$, which leads to equations \ref{estKx} and \ref{estKy}.
\end{proof}
\noindent
We note that the above proof of Proposition \ref{optClustProp} implies that $X$-node groupings are all the same size, and that the $Y$-node groupings are all the same size. This assumption is relaxed in the practical implementation of this methodology we propose: this point is discussed further in Section \ref{practSizeEstSect}.

\subsubsection*{Practical estimation of the number of $X$ and $Y$-node groupings}\label{practSizeEstSect}
We implement spectral clustering by including a standard $k$-means step, to group the $X$- and $Y$-nodes in the spaces of the left and right singular vectors corresponding to the 2\textsuperscript{nd} to $k^{(X)}$\textsuperscript{th} and 2\textsuperscript{nd} to $k^{(Y)}$\textsuperscript{th} greatest singular values, respectively, of the singular value decomposition of the co-Laplacian $\mathbf{L}_{XY}$ (equation \ref{coLaplacian}). This $k$-means step does not produce identical group sizes, however we note that the estimates of $k^{(X)}$ and $k^{(Y)}$ defined according to equations \ref{estKx} and \ref{estKy} assume that the $X$ and $Y$ node groupings are the same size (i.e., that the blocks in the blockmodel are all the same size with respect to the $X$-nodes, and separately with respect to the $Y$-nodes). We relax this requirement in practice (whilst noting that we must maintain $ \min_i h^{(X)}_i/\max_i h^{(X)}_i=\Theta(1)$), because after examining several empirical data-sets of the type presented in the next section, we observed that the group sizes produced by this type of regularised degree-corrected spectral clustering, tend not to vary significantly in size (there are no `giant clusters'). Further, this requirement of identical group sizes is not physically realistic in the practical examples we present in the next section, and in many other real scenarios.

To estimate $\widetilde{M}$ and $\gamma$, we approximate the maximum slope of the graphon separately in the directions corresponding to the $X$ and $Y$-nodes, by considering the top component of the singular value decomposition of the adjacency matrix $\mathbf{A}$. This is equivalent to the rule-of-thumb procedure in the network histogram method, in the symmetric/unipartite community detection scenario \citep{olhede2014net}. The top left and right singular vectors are ordered, and their gradients and values at their midpoints (the expected points of maximum slope) are estimated as $\hat{p}_X$ and $\hat{b}_X$ respectively for the $X$-nodes, and $\hat{p}_Y$ and $\hat{b}_Y$ respectively for the $Y$-nodes. By thinking of this singular value decomposition as a factorisation of the scaled, discretely-sampled graphon (i.e., the ordered adjacency matrix), denoting the greatest singular value as $\nu$, leads to the linear approximations for the maximum gradient of the isotropic graphon $M$ in the directions of the $X$ and $Y$-nodes, $M_X$ and $M_Y$ respectively:
\begin{equation*}
\hat{M}_X=\frac{\nu}{\rho}\hat{p}_X\hat{b}_Ym,\qquad\hat{M}_Y=\frac{\nu}{\rho}\hat{b}_X\hat{p}_Yl,
\end{equation*}
where $m$ and $l$ are the number of $X$ and $Y$-nodes respectively (as previously defined). These factors $m$ and $l$ take account of the fact that the isotropic graphon margins are bounded on $[0,1]$, whereas the adjacency matrix margins take the values $\{1,...,m\}$ and $\{1,...,l\}$, and the edge density factor $\rho$ (defined as in Lemma \ref{MISElem}) normalises with respect to the adjacency matrix realisation, such that the above estimates are independent of edge density $\rho$. The linear stretch transformation $\gamma$ defines the maximum gradients of the equi-smooth graphon as $\widetilde{M}_X=\gamma M_X$ and $\widetilde{M}_Y=M_Y/\gamma $ respectively, and hence an estimate of the squared maximum gradient of the isotropic graphon can be found as:
\begin{equation*}
\hat{\widetilde{M}}^2=\gamma^2\cdot\hat{M_X}^2+\frac{1}{\gamma^2}\cdot\hat{M_Y}^2=\frac{\nu^2}{\rho^2}\left(\gamma^2\cdot\hat{p}_X^2\hat{b}_Y^2m^2+\frac{1}{\gamma^2}\cdot\hat{b}_X^2\hat{p}_Y^2l^2\right).
\end{equation*}
Using the assumption that the equi-smooth graphon is Lipschitz-continuous, with the same upper-bound on its smoothness with respect to both the $X$ and $Y$ nodes, i.e., $\widetilde{M}_X=\widetilde{M}_Y$, $\implies$ $\gamma M_X=M_Y/\gamma$, we can estimate $\gamma$ as:
\begin{equation}
\hat{\gamma}^2=\frac{\hat{M}_Y}{\hat{M}_X}.
\end{equation}

\subsection*{Appendix D: Proof of Lemma \ref{MISElem}}
Define $\mathbf{A}$, $k^{(X)}$, $k^{(Y)}$ according to Definition \ref{mainModelDef}, define $\xi^{(X)}$ and $\xi^{(Y)}$ according to Definition \ref{anisGraphModelDef}, define $f$, $\widetilde{f}$ and $\gamma$ according to Definition \ref{anGraGammaDef}, and define $\rho$ and $\widetilde{M}$ according to Lemma \ref{MISElem}. Define bandwidths $h^{(X)}_p=\left|g^{(X)}_p\right|$ and $h^{(Y)}_q=\left|g^{(Y)}_q\right|$, where $\left|\cdot\right|$ represents cardinality, define $\omega\left(p,q\right)$ as the domain of integration over the block corresponding to the pairing of $g^{(X)}_p$ with $g^{(Y)}_q$ (where $g^{(X)}_p$ and $g^{(Y)}_q$ are sets of $X$-nodes and $Y$ nodes, with $G^{(X)}$ and $G^{(Y)}$ respectively the sets of $g^{(X)}_p$ and $g^{(Y)}_q$ over $p\in\left\{1,...,k^{(X)}\right\}$ and $q\in\left\{1,...,k^{(Y)}\right\}$), and define $\bar{A}_{p,q}$ as the block average corresponding to the pairing of $g^{(X)}_p$ with $g^{(Y)}_q$,
\begin{equation*}
\bar{A}_{p,q}=\frac{\sum_{j\in g^{(Y)}_q}\sum_{i\in g^{(X)}_p}A_{ij}}{h^{(X)}_p\cdot h^{(Y)}_q}.
\end{equation*}
For convenience, we also define here theoretical relations to $\bar{A}_{p,q}$, by denoting the average values of $f$ and $f^2$ over the block corresponding to the pairing of $g^{(X)}_p$ with $g^{(Y)}_q$ as $\bar{f}_{p,q}$ and $\bar{f^2}_{p,q}$ respectively:
\begin{equation}
\bar{f}_{p,q}=\frac{1}{\left|\omega\left(p,q\right)\right|}\iint_{\omega\left(p,q\right)}f(x,y)dx\ dy\label{fBarDef}
\end{equation}
and
\begin{equation}
\bar{f^2}_{p,q}=\frac{1}{\left|\omega\left(p,q\right)\right|}\iint_{\omega\left(p,q\right)}f^2(x,y)dx\ dy,\label{f2BarDef}
\end{equation}
where
\begin{equation*}
\left|\omega\left(p,q\right)\right|=\frac{h^{(X)}_p}{m}\cdot\frac{h^{(Y)}_q}{l}.
\end{equation*}
The bias-variance decomposition of the oracle MISE of the blockmodel approximation of the graphon function $\hat{f}$ can hence be written as \citep{olhede2014net}:
\begin{multline}
\mathrm{MISE\left(\hat{f}\right)}\leq\mathbb{E}\iint_{(0,1)^2}\left|f(x,y)-\hat{f}(x,y)\right|^2dx\ dy=\\\label{MISEmain}
\sum_{q=1}^{k^{(Y)}}\sum_{p=1}^{k^{(X)}}\iint_{\omega\left(p,q\right)}\left\{\left|f(x,y)-\frac{\mathbb{E}\left(\bar{A}_{p,q}\right)}{\rho}\right|^2+\frac{\mathrm{Var}\left(\bar{A}_{p,q}\right)}{\rho^2}\right\}dx\ dy.
\end{multline}
The domain of integration $\omega\left(p,q\right)$ is hence a contiguous region of the graphon, which corresponds to entries of the adjacency matrix which are not necessarily contiguous.

Modelling the equi-smooth graphon $\widetilde{f}$ as a linear stretch transformation of the anisotropic graphon $f$, by anisotropy factor $\gamma$, means that we can write:
\begin{equation*}
f(x,y)=\widetilde{f}\left(\gamma x,y/\gamma\right).
\end{equation*}
We define the graphon oracle \citep{wolfe2013nonparametric,olhede2014net} ordering of the $X$ and $Y$-nodes according to $\xi^{(X)}$ and $\xi^{(Y)}$ respectively. These are unobservable latent random vectors, which map the locations of the $X$ and $Y$ nodes from the margins of the graphon to the margins of the adjacency matrix. I.e., $\xi^{(X)}_i$ and $\xi^{(Y)}_j$ provide the locations on the graphon margins which correspond to the $X$ and $Y$-nodes $i$ and $j$ respectively, where $i$ and $j$ are the adjacency matrix indices of these nodes. We define $(i)^{-1}$ as a function which gives the rank of $\xi^{(X)}_i$, $1\leq i\leq m$, and similarly $(j)^{-1}$ as a function which gives the rank of $\xi^{(Y)}_j$, $1\leq j\leq l$. Therefore, $(i)^{-1}$ and $(j)^{-1}$ are functions which take the ordering along the adjacency matrix margins, and return the ordering along the graphon margins. Hence, the inverses of these functions, $(i)$ and $(j)$, take the ordering along the graphon margins, and return the corresponding ordering along the adjacency matrix margins. Adapting the proof of Lemma 3 from \citep{olhede2014net} to the anisotropic graphon, by defining $i_m=i/(m+1)$ and $j_l=j/(l+1)$, and assuming that $\widetilde{f}$ is Lipschitz-continuous, gives:
\begin{align*}
\left|f\left(\xi^{(X)}_{\left(i\right)},\xi^{(Y)}_{\left(j\right)}\right)-f\left(i_m,j_l\right)\right|=\left|\widetilde{f}\left(\gamma\xi^{(X)}_{\left(i\right)},\xi^{(Y)}_{\left(j\right)}/\gamma\right)-\widetilde{f}\left(\gamma i_m,j_l/\gamma\right)\right|\\[1ex]
\leq \widetilde{M}\left|\left(\gamma\xi^{(X)}_{\left(i\right)},\xi^{(Y)}_{\left(j\right)}/\gamma\right)-\left(\gamma i_m,j_l/\gamma\right)\right|.
\end{align*}
Writing the variances and applying Jensen's inequality as in \citep{olhede2014net} we get,
\begin{align*}
&\mathrm{Var}\left(\xi^{(X)}_{\left(i\right)}\right)=\frac{i_m(1-i_m)}{m+2}\leq\frac{1/4}{m+2},\\[1ex]
&\mathrm{Var}\left(\xi^{(Y)}_{\left(j\right)}\right)=\frac{j_l(1-j_l)}{l+2}\leq\frac{1/4}{l+2},\\
\implies&\mathbb{E}_{\xi^{(X)},\xi^{(Y)}}\left\{\gamma^2\left(\xi^{(X)}_{\left(i\right)}-i_m\right)^2+\frac{1}{\gamma^2}\left(\xi^{(Y)}_{\left(j\right)}-j_l\right)^2\right\}^{\frac{1}{2}}\\
&\quad\leq\left(\gamma^2\mathrm{Var}\left(\xi^{(X)}_{\left(i\right)}\right)+\frac{1}{\gamma^2}\mathrm{Var}\left(\xi^{(Y)}_{\left(j\right)}\right)\right)^{\frac{1}{2}}\\
&\quad\leq\left\{\gamma^2\cdot\frac{1}{4(m+2)}+\frac{1}{\gamma^2}\cdot\frac{1}{4(l+2)}\right\}^{\frac{1}{2}},\\
\therefore&\enskip\mathbb{E}_{\xi^{(X)},\xi^{(Y)}}\left|f\left(\xi^{(X)}_{\left(i\right)},\xi^{(Y)}_{\left(j\right)}\right)-f\left(i_m,j_l\right)\right|\leq\widetilde{M}\left\{\gamma^2\cdot\frac{1}{4(m+2)}+\frac{1}{\gamma^2}\cdot\frac{1}{4(l+2)}\right\}^{\frac{1}{2}}.\numberthis\label{expAt1}
\end{align*}
We note that this explains the stretching factor $\gamma$. Now adapting Lemma 2 from \citep{olhede2014net}, we apply the law of iterated expectations to  $A_{(i)(j)}$, to obtain:
\begin{equation}
\mathbb{E}\left(A_{(i)(j)}\right)=\mathbb{E}_{\xi^{(X)},\xi^{(Y)}}\left[\mathbb{E}_{A|\xi^{(X)},\xi^{(Y)}}\left(A_{(i)(j)}\middle|\xi^{(X)},\xi^{(Y)}\right)\right]=\mathbb{E}_{\xi^{(X)},\xi^{(Y)}}\left[\rho f\left(\xi^{(X)}_{\left(i\right)},\xi^{(Y)}_{\left(j\right)}\right)\right],\label{expAt2}
\end{equation}
then using Jensen's inequality we get:
\begin{equation}
\left|\mathbb{E}_{\xi^{(X)},\xi^{(Y)}}\left[\rho f\left(\xi^{(X)}_{\left(i\right)},\xi^{(Y)}_{\left(j\right)}\right)\right]-\rho f\left(i_m,j_l\right)\right|\leq\rho\mathbb{E}_{\xi^{(X)},\xi^{(Y)}}\left[\left|f\left(\xi^{(X)}_{\left(i\right)},\xi^{(Y)}_{\left(j\right)}\right)-f\left(i_m,j_l\right)\right|\right],\label{expAt3}
\end{equation}
and hence combining equations \ref{expAt1}-\ref{expAt3}, we have:
\begin{equation}
\left|\mathbb{E}\left(A_{(i)(j)}\right)-\rho f\left(i_m,j_l\right)\right|\leq\rho \widetilde{M}\left\{\gamma^2\cdot\frac{1}{4(m+2)}+\frac{1}{\gamma^2}\cdot\frac{1}{4(l+2)}\right\}^{\frac{1}{2}}.\label{lem2exp}
\end{equation}
Now applying the law of total variance to $A_{(i)(j)}$, as in Lemma 2 from \citep{olhede2014net}, we obtain:
\begin{align*}
\text{Var}\left(A_{(i)(j)}\right)=&\mathbb{E}_{\xi^{(X)},\xi^{(Y)}}\left[\text{Var}_{A|\xi^{(X)},\xi^{(Y)}}\left(A_{(i)(j)}\middle|\xi^{(X)},\xi^{(Y)}\right)\right]\\
&+\text{Var}_{\xi^{(X)},\xi^{(Y)}}\left[\mathbb{E}_{A|\xi^{(X)},\xi^{(Y)}}\left(A_{(i)(j)}\middle|\xi^{(X)},\xi^{(Y)}\right)\right]\\
=&\mathbb{E}_{\xi^{(X)},\xi^{(Y)}}\left[\rho f\left(\xi^{(X)}_{\left(i\right)},\xi^{(Y)}_{\left(j\right)}\right)\left(1-\rho f\left(\xi^{(X)}_{\left(i\right)},\xi^{(Y)}_{\left(j\right)}\right)\right)\right]\\
&+\mathbb{E}_{\xi^{(X)},\xi^{(Y)}}\left[\rho^2\left(f\left(\xi^{(X)}_{\left(i\right)},\xi^{(Y)}_{\left(j\right)}\right)\right)^2\right]-\left(\mathbb{E}_{\xi^{(X)},\xi^{(Y)}}\left[\rho f\left(\xi^{(X)}_{\left(i\right)},\xi^{(Y)}_{\left(j\right)}\right)\right]\right)^2\\
=&\mathbb{E}_{\xi^{(X)},\xi^{(Y)}}\left[\rho f\left(\xi^{(X)}_{\left(i\right)},\xi^{(Y)}_{\left(j\right)}\right)\right]-\mathbb{E}_{\xi^{(X)},\xi^{(Y)}}\left[\rho^2\left(f\left(\xi^{(X)}_{\left(i\right)},\xi^{(Y)}_{\left(j\right)}\right)\right)^2\right]\\
&+\mathbb{E}_{\xi^{(X)},\xi^{(Y)}}\left[\rho^2\left(f\left(\xi^{(X)}_{\left(i\right)},\xi^{(Y)}_{\left(j\right)}\right)\right)^2\right]-\left(\mathbb{E}_{\xi^{(X)},\xi^{(Y)}}\left[\rho f\left(\xi^{(X)}_{\left(i\right)},\xi^{(Y)}_{\left(j\right)}\right)\right]\right)^2\\
=&\mathbb{E}_{\xi^{(X)},\xi^{(Y)}}\left[\rho f\left(\xi^{(X)}_{\left(i\right)},\xi^{(Y)}_{\left(j\right)}\right)\right]\left\{\mathbb{E}_{\xi^{(X)},\xi^{(Y)}}\left[1-\rho f\left(\xi^{(X)}_{\left(i\right)},\xi^{(Y)}_{\left(j\right)}\right)\right]\right\}.\numberthis\label{varAt1}
\end{align*}
From equation \ref{expAt1}, we get:
\begin{equation}
\mathbb{E}_{\xi^{(X)},\xi^{(Y)}}\left[\rho f\left(\xi^{(X)}_{\left(i\right)},\xi^{(Y)}_{\left(j\right)}\right)\right]\leq\rho f\left(i_m,j_l\right)+\rho\widetilde{M}\left\{\gamma^2\cdot\frac{1}{4(m+2)}+\frac{1}{\gamma^2}\cdot\frac{1}{4(l+2)}\right\}^{\frac{1}{2}}\label{varAintT1}
\end{equation}
and
\begin{equation}
-\mathbb{E}_{\xi^{(X)},\xi^{(Y)}}\left[\rho f\left(\xi^{(X)}_{\left(i\right)},\xi^{(Y)}_{\left(j\right)}\right)\right]\leq-\rho f\left(i_m,j_l\right)+\rho\widetilde{M}\left\{\gamma^2\cdot\frac{1}{4(m+2)}+\frac{1}{\gamma^2}\cdot\frac{1}{4(l+2)}\right\}^{\frac{1}{2}},\label{varAintT2}
\end{equation}
and hence also
\begin{equation}
\mathbb{E}_{\xi^{(X)},\xi^{(Y)}}\left[1-\rho f\left(\xi^{(X)}_{\left(i\right)},\xi^{(Y)}_{\left(j\right)}\right)\right]\geq1-\rho f\left(i_m,j_l\right)-\rho\widetilde{M}\left\{\gamma^2\cdot\frac{1}{4(m+2)}+\frac{1}{\gamma^2}\cdot\frac{1}{4(l+2)}\right\}^{\frac{1}{2}}\label{varAintT3}
\end{equation}
and
\begin{equation}
-\mathbb{E}_{\xi^{(X)},\xi^{(Y)}}\left[1-\rho f\left(\xi^{(X)}_{\left(i\right)},\xi^{(Y)}_{\left(j\right)}\right)\right]\geq-1+\rho f\left(i_m,j_l\right)-\rho\widetilde{M}\left\{\gamma^2\cdot\frac{1}{4(m+2)}+\frac{1}{\gamma^2}\cdot\frac{1}{4(l+2)}\right\}^{\frac{1}{2}}.\label{varAintT4}
\end{equation}
Now combining equation \ref{varAintT1} with the negative of equation \ref{varAintT4} and applying equation equation \ref{varAt1} we get:
\begin{align*}
\text{Var}\left(A_{(i)(j)}\right)\leq&\left[\rho f\left(i_m,j_l\right)+\rho\widetilde{M}\left\{\gamma^2\cdot\frac{1}{4(m+2)}+\frac{1}{\gamma^2}\cdot\frac{1}{4(l+2)}\right\}^{\frac{1}{2}}\right]\\
&\quad\cdot\left[1-\rho f\left(i_m,j_l\right)+\rho\widetilde{M}\left\{\gamma^2\cdot\frac{1}{4(m+2)}+\frac{1}{\gamma^2}\cdot\frac{1}{4(l+2)}\right\}^{\frac{1}{2}}\right]
\end{align*}
and hence:
\begin{align*}
\text{Var}&\left(A_{(i)(j)}\right)\leq\rho f\left(i_m,j_l\right)\left[1-\rho f\left(i_m,j_l\right)\right]\\
&+\rho\widetilde{M}\left\{\gamma^2\cdot\frac{1}{4(m+2)}+\frac{1}{\gamma^2}\cdot\frac{1}{4(l+2)}\right\}^{\frac{1}{2}}\left[1+\rho\widetilde{M}\left\{\gamma^2\cdot\frac{1}{4(m+2)}+\frac{1}{\gamma^2}\cdot\frac{1}{4(l+2)}\right\}^{\frac{1}{2}}\right].\\\numberthis\label{varAt2}
\end{align*}
Similarly combining the negative of equation \ref{varAintT2} with equation \ref{varAintT3} and applying equation equation \ref{varAt1} we get:
\begin{align*}
\text{Var}\left(A_{(i)(j)}\right)\geq&\left[\rho f\left(i_m,j_l\right)-\rho\widetilde{M}\left\{\gamma^2\cdot\frac{1}{4(m+2)}+\frac{1}{\gamma^2}\cdot\frac{1}{4(l+2)}\right\}^{\frac{1}{2}}\right]\\
&\quad\cdot\left[1-\rho f\left(i_m,j_l\right)-\rho\widetilde{M}\left\{\gamma^2\cdot\frac{1}{4(m+2)}+\frac{1}{\gamma^2}\cdot\frac{1}{4(l+2)}\right\}^{\frac{1}{2}}\right],
\end{align*}
and hence:
\begin{align*}
\text{Var}&\left(A_{(i)(j)}\right)\geq\rho f\left(i_m,j_l\right)\left[1-\rho f\left(i_m,j_l\right)\right]\\
&-\rho\widetilde{M}\left\{\gamma^2\cdot\frac{1}{4(m+2)}+\frac{1}{\gamma^2}\cdot\frac{1}{4(l+2)}\right\}^{\frac{1}{2}}\left[1-\rho\widetilde{M}\left\{\gamma^2\cdot\frac{1}{4(m+2)}+\frac{1}{\gamma^2}\cdot\frac{1}{4(l+2)}\right\}^{\frac{1}{2}}\right],
\end{align*}
and therefore:
\begin{align*}
-\text{Var}&\left(A_{(i)(j)}\right)\leq-\rho f\left(i_m,j_l\right)\left[1-\rho f\left(i_m,j_l\right)\right]\\
&+\rho\widetilde{M}\left\{\gamma^2\cdot\frac{1}{4(m+2)}+\frac{1}{\gamma^2}\cdot\frac{1}{4(l+2)}\right\}^{\frac{1}{2}}\left[1-\rho\widetilde{M}\left\{\gamma^2\cdot\frac{1}{4(m+2)}+\frac{1}{\gamma^2}\cdot\frac{1}{4(l+2)}\right\}^{\frac{1}{2}}\right]\\\numberthis\label{varAt3}
\leq&-\rho f\left(i_m,j_l\right)\left[1-\rho f\left(i_m,j_l\right)\right]\\
&+\rho\widetilde{M}\left\{\gamma^2\cdot\frac{1}{4(m+2)}+\frac{1}{\gamma^2}\cdot\frac{1}{4(l+2)}\right\}^{\frac{1}{2}}\left[1+\rho\widetilde{M}\left\{\gamma^2\cdot\frac{1}{4(m+2)}+\frac{1}{\gamma^2}\cdot\frac{1}{4(l+2)}\right\}^{\frac{1}{2}}\right],
\end{align*}
and hence combining equations \ref{varAt2} and \ref{varAt3} we get:
\begin{multline}
\left|\mathrm{Var}\left(A_{(i)(j)}\right)-\rho f\left(i_m,j_l\right)\left[1-\rho f\left(i_m,j_l\right)\right]\right|\\\label{lem2var}
\leq\rho\widetilde{M}\left\{\gamma^2\cdot\frac{1}{4(m+2)}+\frac{1}{\gamma^2}\cdot\frac{1}{4(l+2)}\right\}^{\frac{1}{2}}\cdot\left[1+\rho \widetilde{M}\left\{\gamma^2\cdot\frac{1}{4(m+2)}+\frac{1}{\gamma^2}\cdot\frac{1}{4(l+2)}\right\}^{\frac{1}{2}}\right].
\end{multline}
Now referring to equation \ref{lem2exp} and comparing it to equation 6 of Supporting Information Section A in \citep{olhede2014net}, allows us to re-write the covariance expression in Lemma 2 of \citep{olhede2014net} giving:
\begin{equation}
\mathrm{Cov}\left(A_{(i)(j)},A_{(i')(j')}\right)\leq\rho^2\widetilde{M}^2\left\{\gamma^2\cdot\frac{1}{4(m+2)}+\frac{1}{\gamma^2}\cdot\frac{1}{4(l+2)}\right\}\label{lem2cov},
\end{equation}
$i\neq i'$, $j\neq j'$.
We can then use equations \ref{lem2exp}, \ref{lem2var} and \ref{lem2cov} to adapt Proposition 1 from \citep{olhede2014net}, also referring to equations \ref{fBarDef} and \ref{f2BarDef}, to give:
\begin{equation}
\left|\mathbb{E}\left(\bar{A}_{p,q}\right)-\rho\bar{f}_{p,q}\right|\leq\rho \widetilde{M}\left\{\gamma^2\cdot\frac{1}{4m}+\frac{1}{\gamma^2}\cdot\frac{1}{4l}\right\}^{\frac{1}{2}}\left\{1+o\left(1\right)\right\}\label{OWprop1exp}
\end{equation}
and
\begin{multline}
\left|\mathrm{Var}\left(\bar{A}_{p,q}\right)-\frac{\rho\bar{f}_{p,q}-\rho^2\bar{f^2}_{p,q}}{h^{(X)}_p\cdot h^{(Y)}_q}\right|\\[2ex]
\leq\frac{\rho \widetilde{M}}{h^{(X)}_p\cdot h^{(Y)}_q}\left\{\gamma^2\cdot\frac{1}{4m}+\frac{1}{\gamma^2}\cdot\frac{1}{4l}\right\}^{\frac{1}{2}}\left\{1+o\left(1\right)\right\}+\rho^2\widetilde{M}^2\left\{\gamma^2\cdot\frac{1}{4m}+\frac{1}{\gamma^2}\cdot\frac{1}{4l}\right\},\label{OWprop1var}
\end{multline}
which is a conservative upper bound. Now substituting equation \ref{OWprop1var} back into equation \ref{MISEmain}, we get:
\begin{multline*}
\mathrm{MISE\left(\hat{f}\right)}\leq\sum_{q=1}^{k^{(Y)}}\sum_{p=1}^{k^{(X)}}\iint_{\omega\left(p,q\right)}\left[\left|\left\{f(x,y)-\bar{f}_{p,q}\right\}\right.\right.\\
\left. +\left\{\bar{f}_{p,q}-\mathbb{E}\left(\bar{A}_{p,q}\right)/\rho\right\}\right|^2+\frac{\bar{f}_{p,q}-\rho\bar{f^2}_{p,q}}{\rho\cdot h^{(X)}_p\cdot h^{(Y)}_q}\\
+\left. \frac{\widetilde{M}}{\rho\cdot h^{(X)}_p\cdot h^{(Y)}_q}\left\{\gamma^2\cdot\frac{1}{4m}+\frac{1}{\gamma^2}\cdot\frac{1}{4l}\right\}^{\frac{1}{2}}\left\{1+o\left(1\right)\right\}+\widetilde{M}^2\left\{\gamma^2\cdot\frac{1}{4m}+\frac{1}{\gamma^2}\cdot\frac{1}{4l}\right\}\right]dx\ dy,
\end{multline*}
then substituting equation \ref{OWprop1exp}, integrating and rearranging, leads to:
\begin{multline}
\mathrm{MISE\left(\hat{f}\right)}\leq\sum_{q=1}^{k^{(Y)}}\sum_{p=1}^{k^{(X)}}\left[\iint_{\omega\left(p,q\right)}\left|f(x,y)-\bar{f}_{p,q}\right|^2dx\ dy\right.\\\label{MISEintermediate}
+\left(2\widetilde{M}^2\left\{\gamma^2\cdot\frac{1}{4m}+\frac{1}{\gamma^2}\cdot\frac{1}{4l}\right\}\left\{1+o\left(1\right)\right\}+\frac{\bar{f}_{p,q}-\rho\bar{f^2}_{p,q}}{\rho\cdot h^{(X)}_p\cdot h^{(Y)}_q}\right. \\
+\left.\left. \frac{\widetilde{M}}{\rho\cdot h^{(X)}_p\cdot h^{(Y)}_q}\left\{\gamma^2\cdot\frac{1}{4m}+\frac{1}{\gamma^2}\cdot\frac{1}{4l}\right\}^{\frac{1}{2}}\left\{1+o\left(1\right)\right\}\right)\cdot\frac{h^{(X)}_p}{m}\cdot\frac{h^{(Y)}_q}{l}\right].
\end{multline}
Then, adapting the proof of Lemma 1 from \citep{olhede2014net}, we can write:
\begin{align*}
\left|\bar{f}_{p,q}-f(x,y)\right|=\left|\frac{1}{\left|\omega\left(p,q\right)\right|}\iint_{\omega\left(p,q\right)}f(x',y')dx'\ dy'-f(x,y)\right|\\[2ex]
\leq\frac{1}{\left|\omega\left(p,q\right)\right|}\iint_{\omega\left(p,q\right)}\left|\widetilde{f}(\gamma x',y'/\gamma)-\widetilde{f}(\gamma x,y/\gamma)\right|dx'\ dy'.
\end{align*}
Assuming $\widetilde{f}$ is Lipschitz continuous, it therefore follows that:
\begin{align*}
\left|\bar{f}_{p,q}-f(x,y)\right|\leq\frac{1}{\left|\omega\left(p,q\right)\right|}\iint_{\omega\left(p,q\right)}\widetilde{M}\left|(\gamma x',y'/\gamma)-(\gamma x,y/\gamma)\right|dx'\ dy'\\
\leq\frac{1}{\left|\omega\left(p,q\right)\right|}\iint_{\omega\left(p,q\right)}\widetilde{M}\sqrt{\gamma^2\cdot\frac{\left(h^{(X)}_p\right)^2}{m^2}+\frac{1}{\gamma^2}\cdot\frac{\left(h^{(Y)}_q\right)^2}{l^2}}dx'\ dy'
\end{align*}
\begin{equation*}
\implies\left|\bar{f}_{p,q}-f(x,y)\right|\leq \widetilde{M}\sqrt{\gamma^2\cdot\frac{\left(h^{(X)}_p\right)^2}{m^2}+\frac{1}{\gamma^2}\cdot\frac{\left(h^{(Y)}_q\right)^2}{l^2}}
\end{equation*}
and therefore
\begin{equation*}
\frac{1}{\left|\omega\left(p,q\right)\right|}\iint_{\omega\left(p,q\right)}\left|\bar{f}_{p,q}-f(x,y)\right|^2\leq \widetilde{M}^2\left\{\gamma^2\cdot\frac{\left(h^{(X)}_p\right)^2}{m^2}+\frac{1}{\gamma^2}\cdot\frac{\left(h^{(Y)}_q\right)^2}{l^2}\right\},
\end{equation*}
and hence summing over all the blocks corresponding to all pairings of $X$-node groupings $g^{(X)}\in G^{(X)}$ with $Y$-node groupings $g^{(Y)}\in G^{(Y)}$, and assuming $h^{(X)}$ and $h^{(Y)}$ are both constants, we get:
\begin{equation}
\sum_{q=1}^{k^{(Y)}}\sum_{p=1}^{k^{(X)}}\iint_{\omega\left(p,q\right)}\left|\bar{f}_{p,q}-f(x,y)\right|^2\leq \widetilde{M}^2\left\{\gamma^2\cdot\frac{\left(h^{(X)}\right)^2}{m^2}+\frac{1}{\gamma^2}\cdot\frac{\left(h^{(Y)}\right)^2}{l^2}\right\}.\label{OWlem1}
\end{equation}
Recalling equation \ref{fBarDef} and equation \ref{graphonIntConstr}, i.e., 
\begin{equation*}
\iint_{(0,1)^2}f(x,y)dx\ dy=1,
\end{equation*}
and noting that:
\begin{equation*}
\sum_{q=1}^{k^{(Y)}}\sum_{p=1}^{k^{(X)}}\frac{\bar{f}_{p,q}-\rho\bar{f^2}_{p,q}}{\rho\cdot h^{(X)}_p\cdot h^{(Y)}_q}\leq\sum_{q=1}^{k^{(Y)}}\sum_{p=1}^{k^{(X)}}\frac{\bar{f}_{p,q}}{\rho\cdot h^{(X)}\cdot h^{(Y)}},
\end{equation*}
we can see that:
\begin{align}
\label{blockSumSimp}\sum_{q=1}^{k^{(Y)}}\sum_{p=1}^{k^{(X)}}\frac{\bar{f}_{p,q}-\rho\bar{f^2}_{p,q}}{\rho\cdot h^{(X)}_p\cdot h^{(Y)}_q}\leq&\sum_{q=1}^{k^{(Y)}}\sum_{p=1}^{k^{(X)}}\frac{m\cdot l}{\rho\cdot \left(h^{(X)}\right)^2\cdot\left(h^{(Y)}\right)^2}\cdot\frac{h^{(X)}}{m}\cdot\frac{h^{(Y)}}{l}\cdot\bar{f}_{p,q}\\[1ex]\nonumber
&=\frac{m\cdot l}{\rho\cdot\left(h^{(X)}\right)^2\cdot\left(h^{(Y)}\right)^2}\sum_{q=1}^{k^{(Y)}}\sum_{p=1}^{k^{(X)}}\iint_{\omega\left(p,q\right)}f(x,y)dx\ dy\\\nonumber
&=\frac{m\cdot l}{\rho\cdot\left(h^{(X)}\right)^2\cdot\left(h^{(Y)}\right)^2}\iint_{(0,1)^2}f(x,y)dx\ dy\\\nonumber
&=\frac{m\cdot l}{\rho\cdot\left(h^{(X)}\right)^2\cdot\left(h^{(Y)}\right)^2}.\nonumber
\end{align}
Now substituting \ref{OWlem1} and \ref{blockSumSimp} into \ref{MISEintermediate}, and rearranging, we get:
\begin{multline*}
\mathrm{MISE\left(\hat{f}\right)}\leq \widetilde{M}^2\left\{\gamma^2\cdot\frac{\left(h^{(X)}\right)^2}{m^2}+\frac{1}{\gamma^2}\cdot\frac{\left(h^{(Y)}\right)^2}{l^2}\right\}\\+2\widetilde{M}^2\left\{\gamma^2\cdot\frac{1}{4m}+\frac{1}{\gamma^2}\cdot\frac{1}{4l}\right\}\left\{1+o\left(1\right)\right\}\\
+\frac{1}{\rho\cdot h^{(X)}\cdot\ h^{(Y)}}+\frac{\widetilde{M}}{\rho\cdot h^{(X)}\cdot h^{(Y)}}\left\{\gamma^2\cdot\frac{1}{4m}+\frac{1}{\gamma^2}\cdot\frac{1}{4l}\right\}^{\frac{1}{2}}\left\{1+o\left(1\right)\right\}
\end{multline*}
and hence:
\begin{multline*}
\mathrm{MISE\left(\hat{f}\right)}\leq \widetilde{M}^2\left\{\gamma^2\cdot\frac{\left(h^{(X)}\right)^2}{m^2}+\frac{1}{\gamma^2}\cdot\frac{\left(h^{(Y)}\right)^2}{l^2}\right\}\\+2\widetilde{M}^2\left\{\gamma^2\cdot\frac{1}{4m}+\frac{1}{\gamma^2}\cdot\frac{1}{4l}\right\}\left\{1+o\left(1\right)\right\}+\frac{1}{\rho\cdot h^{(X)}\cdot\ h^{(Y)}}\left\{1+o\left(1\right)\right\} .
\end{multline*}
\label{MISEsubbed}

\end{document}